\documentclass[conference,compsoc]{IEEEtran}

\usepackage[utf8]{inputenc}
\usepackage[english]{babel}
\usepackage[T1]{fontenc}
\PassOptionsToPackage{svgnames,x11names,table}{xcolor}
\usepackage{hyphenat}
\usepackage{mathrsfs}
\usepackage{microtype}

\usepackage{hyperref}
\hypersetup{
	colorlinks,
	linkcolor={black},
	citecolor={red!70!black},
	urlcolor={blue!70!black}
}

\usepackage{cite}
\usepackage{amsmath,amssymb,amsfonts}
\usepackage{algorithmic}
\usepackage{graphicx}
\usepackage{textcomp}
\usepackage{xcolor}

\usepackage{tikz}

\usepackage{siunitx}
\usepackage{comment}
\usepackage[normalem]{ulem}
\usepackage[capitalise, noabbrev]{cleveref}
\usepackage{multirow}
\usepackage{dsfont}
\usepackage{booktabs}
\usepackage{xspace}
\usepackage{bbold}
\usepackage{caption}
\usepackage{subcaption}
\usepackage{enumitem}
\usepackage{amsthm}

\newtheorem{lemma}{Lemma}

\newtheorem{definition}{Definition}
\newtheorem{theorem}{Theorem}
\newtheorem{corollary}{Corollary}

\newcommand{\etal}{\textit{et al.}}

\newcommand{\prob}[1]{\Pr \left[ #1 \right]}

\newcommand{\ExpVal}[2]{\mathbb{E}_{#1}\left[ #2 \right]}

\newcommand{\miner}{\ensuremath{v}\xspace}
\newcommand{\nodes}{\ensuremath{n}\xspace}  
\newcommand{\setminers}{\ensuremath{\mathcal{V}}\xspace} 
\newcommand{\setminershon}{\ensuremath{\setminers_{\mathrm{hon}}}\xspace} 
\newcommand{\setminersadv}{\ensuremath{\setminers_{\mathrm{adv}}}\xspace}

\newcommand{\minerprime}{\ensuremath{\miner'}\xspace}
 
\newcommand{\absmp}{\ensuremath{\pi}\xspace} 
 
\newcommand{\zerominers}{zero mining-power nodes}

\newcommand{\abspowerhon}{\ensuremath{\pi_{\mathrm{hon}}}\xspace}  
\newcommand{\abspoweradv}{\ensuremath{\pi_{\mathrm{adv}}}\xspace}

\newcommand{\powerhon}{\ensuremath{\rho_{\mathrm{hon}}}\xspace}  
\newcommand{\poweradv}{\ensuremath{\rho_{\mathrm{adv}}}\xspace}  
\newcommand{\poweradvmax}{\ensuremath{\poweradv^{max}}\xspace}

\newcommand{\delay}{\ensuremath{\delta}\xspace}  
\newcommand{\maxdelay}{\ensuremath{\Delta}\xspace}  
\newcommand{\avgdelay}{\ensuremath{\widehat{\delta}}\xspace}

\newcommand{\rate}{\ensuremath{\lambda}\xspace}  
\newcommand{\RVabs}{\ensuremath{X}\xspace} 
\newcommand{\RVrel}{\ensuremath{R}\xspace} 
 
\newcommand{\RVreladv}{\ensuremath{\RVrel_{\mathrm{adv}}}\xspace} 
\newcommand{\EV}{\ensuremath{p^*}\xspace}

\newcommand{\diameter}{\ensuremath{D}\xspace}

\newcommand{\outdegree}{\ensuremath{d_{out}}\xspace}

\newcommand{\crprobunf}{\ensuremath{p}\xspace}

\newcommand{\charac}{\ensuremath{C}\xspace}

\newcommand{\charprobunf}{\ensuremath{q}\xspace}

\newcommand{\Txtminers}{Validators\xspace}
\newcommand{\txtminers}{validators\xspace}
\newcommand{\txtminer}{validator\xspace}
\newcommand{\txtnodes}{nodes\xspace}
\newcommand{\txtnode}{node\xspace}
\newcommand{\txtzp}{zero-power nodes\xspace}

\newcommand{\numminers}{\ensuremath{n_{\mathrm{val}}}\xspace}

\newcommand{\numzp}{\ensuremath{n_{\mathrm{ZP}}}\xspace}
\newcommand{\numnodes}{\ensuremath{n}\xspace}
\newcommand{\secevent}{\ensuremath{\mathds{E}_{\mathrm{sec}}}\xspace}

\newcommand{\corrupted}{corrupted\xspace} 
\newcommand{\corrupt}{corrupt\xspace} 
\newcommand{\corruption}{corruption\xspace} 
\newcommand{\constDelay}{\ensuremath{\delay_{\mathrm{max}}}\xspace}

\newcommand{\varMainTheorem}{\ensuremath{\Gamma}\xspace} %

\renewcommand{\paragraph}[1]{\par\medskip\noindent\textbf{#1.}\hspace{1ex minus .1ex}}

\begin{document}

\title{Larger-scale Nakamoto-style Blockchains Don't Necessarily Offer Better Security}

\author{
	\IEEEauthorblockN{Jannik Albrecht$^1$, Sebastien Andreina$^2$, Frederik Armknecht$^3$, \\Ghassan Karame$^1$, Giorgia Marson$^2$, Julian Willingmann$^1$}
	\IEEEauthorblockA{
		\textit{$^1$Ruhr University Bochum}, \textit{$^2$NEC Labs Europe, Germany}, 
		\textit{$^3$University of Mannheim Germany}\\
		$^1$\{jannik.albrecht, ghassan.karame, julian.willingmann\}@ruhr-uni-bochum.de\\
		$^2$\{sebastien.andreina, giorgia.marson\}@neclab.eu, 
		$^3$armknecht@uni-mannheim.de
	}
}

\maketitle

\begin{abstract}

Extensive research on Nakamoto-style consensus protocols has shown that network delays degrade the security of these protocols. Established results indicate that, perhaps surprisingly, maximal security is achieved when the network is as small as two nodes due to increased delays in larger networks. This contradicts the very foundation of blockchains, namely that decentralization improves security.

In this paper, we take a closer look at how the network scale affects security of Nakamoto-style blockchains. We argue that a crucial aspect has been neglected in existing security models: the larger the network, the harder it is for an attacker to control a significant amount of power.
To this end, we introduce a probabilistic corruption model to express the increasing difficulty for an attacker to corrupt resources in larger networks. Based on our model, we analyze the impact of the number of nodes on the (maximum) network delay and the fraction of adversarial power. In particular, we show that (1) increasing the number of nodes eventually violates security, but (2) relying on a small number of nodes does not provide decent security provisions either.
We then validate our analysis by means of an empirical evaluation emulating hundreds of thousands of nodes in deployments such as Bitcoin, Monero, Cardano, and Ethereum Classic. Based on our empirical analysis, we concretely analyze the impact of various real-world parameters and configurations on the consistency bounds in existing deployments and on the adversarial power that can be tolerated while providing security. As far as we are aware, this is the first work that analytically and empirically explores the real-world tradeoffs achieved by current popular Nakamoto-style deployments.
\end{abstract}

\section{Introduction}

Over the past decade,  blockchain platforms proliferated and their market shares have grown immensely~\cite{a16zReport}.
Beyond 
enabling the realization of decentralized Internet-scale cryptocurrencies, blockchains provide transparency, strong security guarantees, and enable open access and participation.
It is indeed expected that the blockchain technology will stimulate 
innovation and will positively impact the digital experience of many enterprises around the globe, e.g., in applications for payments, Central-Bank Digital Currency, insurances, crowd funding, or supply chains.

One of the main motivations behind  blockchain platforms is to encode business logic and realize services in a decentralized manner~\cite{Ethereum:dapps}.
The hope is to remove the need for trusting few central entities and hence to increase resilience.
Consequently, the security of blockchain platforms has received considerable scrutiny in the literature---with many attacks showing the limits of foresight in designing Nakamoto-style blockchains~\cite{DBLP:conf/uss/HeilmanKZG15,DBLP:conf/ccs/GervaisRKC15}.

Introduced with Bitcoin~\cite{nakamoto2009bitcoin} and later adopted by other block\-chains, Nakamoto-style consensus refers to a general design of permissionless consensus.
It typically consists of a probabilistic leader-election protocol, e.g., Proof of Work (PoW) or Proof of Stake (PoS), combined with a set of rules for resolving forks (e.g., the longest chain rule).
Intuitively, Nakamoto-style blockchains are designed to guarantee agreement on the order of all events as long as the majority of resources (computing power in PoW and cryptocurrency stake in PoS) is held by honest nodes. Honest majority alone, however, is not sufficient for security: Nakamoto-style consensus additionally requires that blocks are propagated among honest nodes within a certain time interval (known to the protocol). In other words, Nakamoto-style consensus requires a synchronous network.

Recently, a number of contributions~\cite{DBLP:conf/eurocrypt/GarayKL15, DBLP:journals/iacr/KifferRS22, DBLP:conf/icdcs/ZhaoTLWLX20, DBLP:journals/corr/abs-2203-06357, DBLP:conf/ccs/DemboKTTVWZ20} have formally investigated (and improved) the necessary and sufficient conditions on the honest-majority assumption under which the security of Nakamoto-style consensus can be achieved. These contributions have significantly advanced the community's understanding of proving security in existing blockchains and expressed how security depends on different parameters such as the stale block rate (i.e., the rate at which mined blocks are discarded due to a race condition), the maximum network delay, or the fraction of adversarial power. 
However, they only partially capture the initial idea behind decentralization as they do not address how the scale of the network (e.g., the number \nodes of nodes) impacts the security of the platform. 
More precisely, while it is folklore that an increase of \nodes results in an increase in the network delay, we do not know the extent to which security is affected by these delays. Naturally, higher network delays negatively impact the security of Nakamoto-style blockchains. On the other hand, enhancing decentralization intuitively yields an increase in security, as individual nodes control less relative power.  Unfortunately, existing security models commonly assume that the adversarial power is independent of \nodes,  
neglecting the positive impact of the network scale on the adversarial power.  
\emph{In other words, out of these two opposing forces only the former one has been well studied.} 
A further shortcoming is that existing results abstract away various peculiarities that impact real-world deployment models. For example, simply assuming a synchronous network (as is the case in existing analyses) neglects the fact that different blockchain deployments, e.g., Bitcoin and Ethereum Classic, rely on different gossip protocols, which exhibit drastically different robustness to adversarial attacks.

In this paper, our ultimate goal is to analytically and empirically understand the impact of  
{relevant} real-world parameters (such as the number of nodes in the network, the delays in the network, and the power controlled by the adversary) on the security of current Nakamoto-style deployments. To do so, we revisit and extend  
the state-of-the-art security model for blockchains~\cite{DBLP:conf/ccs/DemboKTTVWZ20} to investigate the impact of the number of nodes on the (maximum) network delay and the fraction of adversarial power. 
By leveraging this analysis and our aforementioned security model, we concretely analyze the impact of various real-world parameters on the consistency bounds and adversarial power that can be tolerated in existing {Nakamoto-style} deployments---while encompassing aspects such as the gossip layer and possible network-layer attacks. 
We summarize our contributions in this work as follows:

\begin{figure}
	\centering
	\includegraphics[width=\linewidth]{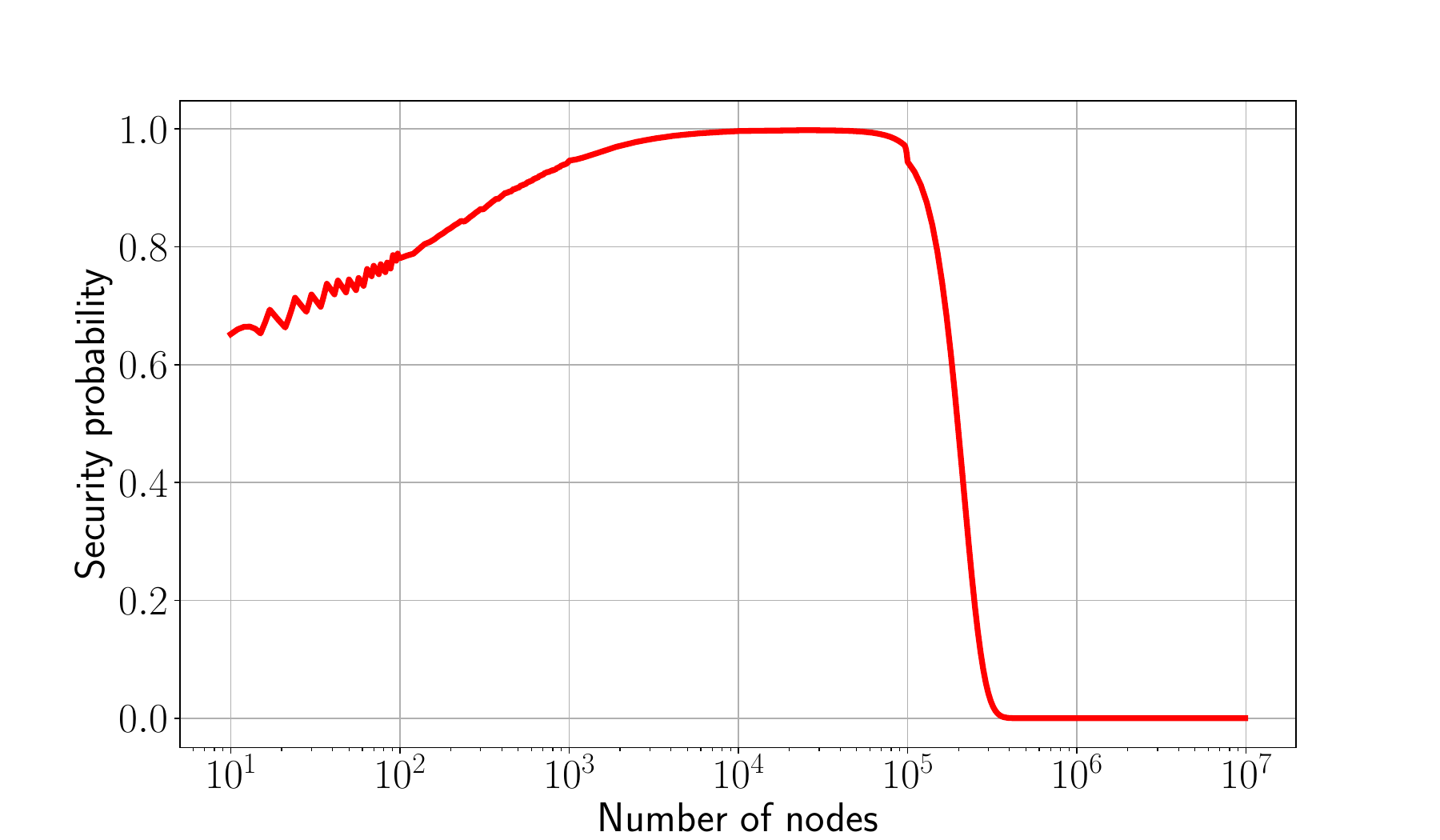}
	\caption{\textit{Security probability} vs number of nodes. We measure the probability (cf. \cref{eq:success}) that a Nakamoto-style blockchain provides persistence and liveness when faced with the probabilistic node corruption. Here, we consider a PoW-based blockchain that generates blocks at a rate of one block per 20 seconds (the case of Cardano). We assume that an adversary can delay selective blocks by up to 100 seconds, i.e., five block generations, (e.g., using~\cite{DBLP:conf/ccs/GervaisRKC15,DBLP:conf/uss/HeilmanKZG15}) and can corrupt each individual miner with probability $12.5\%$. Our results show that a network comprising of \num{20000} nodes is optimal in terms of security, but that a blockchain comprising of 10 nodes offers larger security compared to a network comprising of a million nodes.} 
	\label{fig:motivation}
\end{figure}

\begin{description}
	\item[Gap in existing security models] We show a gap in existing security models for Nakamoto-style blockchains (cf. Section~\ref{sec:modelsec}). Namely, existing models assume an attacker who is able to \corrupt a constant fraction of \txtminers. \emph{This, however, ignores the fact that the number of nodes affects the system's robustness against \corruption.}  We address this gap by revisiting and extending existing security models to capture this very notion that is intrinsic to decentralization.
	\smallskip\item[Theoretical analysis] We establish, for the first time, the relationship between the number of nodes populating a blockchain platform and the resulting adversarial power that can be tolerated by Nakamoto-style consensus (cf. Section~\ref{sec:decentralized}). Our analysis confirms the intuition that an increase of the number of nodes has positive (stronger resilience) and negative (longer delays) aspects. Moreover, it shows that for a small number of nodes, the positive aspects outweigh the negative ones, but only up to a certain threshold. Beyond this threshold, security starts to decrease again for an increasing number of nodes (see example in Figure~\ref{fig:motivation}). 
\smallskip\item[Large-scale evaluation] We validated our findings by means of experiments. Namely, we empirically evaluated the impact of the number of nodes, and other real-world factors, such as the gossip protocol, on the consistency achieved by  
{popular} Nakamoto-style blockchains,  
{namely} Bitcoin, Ethereum Classic, Monero, and Cardano.  
 We also evaluated the impact of  deployments of mining pools with dedicated high-speed network, as well as the impact of real-world network layer attacks on  
 practical deployments of these protocols
 and their respective gossip layers. 
 {The combined outcomes} of our theoretical and empirical analyses  
 {provide new insights towards understanding} the impact of 
 real-world parameters, such as the number of \txtminers, the numbers of \txtzp, and the underlying gossip protocols on the security of blockchain systems (cf.~\Cref{sec:takeaways}).
	For instance, we show that an adversary that controls as little as 28\% and 22\% of the computing power in Cardano and Ethereum Classic (respectively) could violate the consistency in the system when the network comprises a million nodes.
As far as we are aware, these are the first experiments that quantify the impact of large network scales, ranging to hundreds of thousands of nodes, on security. 
For reproducibility purposes, we provide the modified simulator open source\footnote{\url{https://github.com/RUB-InfSec/simblock}}.
\end{description}

\section{System and Security Model}\label{sec:modelsec}

We start by explaining our system model (\cref{sec:system-model}), and overview prior security models (\cref{sec:hma_tresholds}). Finally, we outline gaps in those models (\cref{sec:gaps}).  

\begin{table}	
	\centering
	\footnotesize
	\begin{tabular}{|ll|}
		\hline
		\numminers: 	& number of \txtminers \\
		\numzp: 		& number of \txtzp \\
		\numnodes: 		& number of \txtnodes (\numnodes := \numminers + \numzp) \\
		\poweradv: 		& fraction of adversarial power \\
		\powerhon: 		& fraction of honest power \\
		\rate: 			& block frequency \\
		\maxdelay: 		& upper bound on a blockchain's internal delays \\
		$e$: 			& magnification factor \\
		\hline
	\end{tabular}
	\caption{Summary of parameters used throughout this work.}
	\label{tab:notations}
\end{table}

\subsection{System Model}\label{sec:system-model}
A blockchain is a distributed data-structure, forming a sequence of blocks, where new blocks are appended at a regular pace. 
A \emph{blockchain platform} is composed of a blockchain and a set of participating \txtminers, who form a peer-to-peer (P2P) network and maintain the
blockchain. 
A peer-to-peer network can be modeled by a graph, whose nodes represent the different parties and whose edges represent direct connections. 
In the sequel, we use the term \emph{node} to refer to any of the parties.

The high-level goal of a blockchain platform is to consistently maintain a distributed database of totally-ordered transactions, the blockchain, which is replicated among the various parties and extended by appending new blocks. 
As discussed, we assume here the use of a Nakamoto-style consensus protocol and abstract away from the peculiarities of the underlying leader election protocol (e.g., PoW or PoS). Here, the protocol can be thought of as executing in
synchronous rounds, where each round corresponds to a sufficiently long (compared to the network delay) time interval. 
After retrieving new
blocks and transactions from the network, \txtminers verify the validity
of transactions as well as of possible alternative chains they
received. We remain deliberately vague about how transactions are verified since it is a protocol-specific detail that would otherwise restrict our treatment to specific platforms. Transactions that are deemed correct are included in their next block proposals. When alternative chains are formed, \txtminers then choose which chain they want to accept based on some pre-defined fork resolution protocol (e.g., the longest chain rule in Bitcoin).
Popular Nakamoto-style blockchains include Bitcoin,                                                                                                                                                       Monero, Ethereum Classic and Cardano.

The notion of security in a blockchain-based consensus protocol is 
tightly related to the process of appending blocks (as we explain in detail in \cref{sec:hma_tresholds}). 
Consequently, we follow existing work and explicitly distinguish between nodes that actively participate in the process of generating new blocks (\txtminers) and nodes that simply relay information and blocks (\txtzp). 
The latter are often referred to as \emph{full nodes}: they typically store a copy of the entire ledger and help in verifying transactions but are not involved in generating new blocks. 
We denote by \setminers the set of all \txtminers and refer to the number of \txtminers as $\numminers := \vert\setminers\vert$.
Similarly, we refer to the number of \txtzp as \numzp and the total number of \txtnodes as $\numnodes := \numminers + \numzp$. 
Occasionally, we assume some fixed ordering on \setminers  and simply use $i \in [\numminers]$ for the $i$-th \txtminer in \setminers. 

We denote by $\absmp$ each \txtminer's absolute power, e.g., the number of hashes per second for PoW blockchains or the amount of stake in PoS blockchains, and define $\absmp(\setminers'):=\sum_{\miner\in\setminers'}\absmp(\miner)$ for each (sub-)set $\setminers'$ of \txtminers. 
We assume that \txtminers are working independently. 
That is, if in practice several \txtminers join a pool, then in our model they are represented as one \txtminer. We summarize the various notations that we will use throughout this paper in Table~\ref{tab:notations}.

	\paragraph{Relative power \poweradv} A \txtminer can be either honest, if it follows the protocol specification, or be under adversarial control, i.e., being \corrupted. 
	We denote by \setminershon and \setminersadv the set of honest and adversarial \txtminers, respectively. 
	These sets are disjoint by definition, i.e., $\setminers=\setminershon \dot\cup \setminersadv$. 
	We denote the absolute power of the respective sets by $\abspowerhon=\absmp(\setminershon)$ and $\abspoweradv=\absmp(\setminersadv)$. 
	Similarly, we compute the relative power as follows:
	\begin{equation}\label{eq:rel_power_def}
		\powerhon=\frac{\abspowerhon}{\absmp(\setminers)}\quad\mathrm{and}\quad\poweradv=\frac{\abspoweradv}{\absmp(\setminers)}.
	\end{equation}
\paragraph{Block Generation Rate \rate} \Txtminers generate blocks with a frequency of \rate blocks per second. 
For example, $\rate = \frac{1}{600}$ in Bitcoin, with a new block being found roughly every ten minutes. 
If not stated otherwise, we assume that \rate is a fixed system parameter and in particular independent of \numminers.  We ground this assumption on the dynamic difficulty adjustment of popular Nakamoto-style blockchains, such as Bitcoin and Ethereum 1.0. 
For instance, Bitcoin regularly updates its difficulty to ensure that $\rate = \frac{1}{600}$ for the last 2016 generated blocks. Similarly, Ethereum 1.0 updates the difficulty to ensure that  $\rate = \frac{1}{13}$. 
\paragraph{Maximum Communication Delay \maxdelay}
Conforming with existing Nakamoto-style blockchains, \txtminers exchange messages using a peer-to-peer network. While in practice messages consist of blocks, transactions, and other meta-data, we focus mostly in our analysis on the propagation of blocks; these are typically bulky in size and require non-negligible time to be delivered.
We denote the delay between two communicating \txtminers $\miner_1, \miner_2 \in \setminers$ by $\delay(\miner_1, \miner_2)$, with $\delay(\miner,\miner)=0$ for every \txtminer $\miner\in\setminers$.
Here, the delay expresses how much it takes, in time units, for a block to be propagated from $\miner_1$ to $\miner_2$. In most cases, $\miner_1$ and $\miner_2$ are not communicating directly but use multi-hop-communication. That is, $\delay(\miner_1, \miner_2)$ scales linearly with the length of the shortest path connecting $\miner_1$ and $\miner_2$ in the network.
Adopting common notation, the maximum delay \maxdelay refers to the largest delay between any two nodes:
\begin{equation}
	\maxdelay:=\max_{\miner,\minerprime\in\setminers}\{\delay(\miner,\minerprime)\}.
\end{equation}

\subsection{Existing Security Models} \label{sec:hma_tresholds}
\newcommand{\hmafactor}{\ensuremath{\epsilon}\xspace}

\begin{table*}
	\scriptsize
	\centering
\scalebox{1.0}{\begin{tabular}{ lccl } 
	\toprule
	Authors & Condition & RB / CT & Comments \\ 
	\midrule
	Garay et al.~\cite{DBLP:conf/eurocrypt/GarayKL15} & $\powerhon > \poweradv$ & RB & Delivery of messages at the end of each round, i.e., $\delay_r = 0$.  \\ 
	Pass et al.~\cite{DBLP:conf/eurocrypt/PassSS17} & $\powerhon(1-(2\delay_r + 2))\powerhon) > \poweradv$ & RB & Delivery of messages within $\delay_r$ rounds. \\ 
 	Kiffer et al.~\cite{DBLP:journals/iacr/KifferRS22} &  Non-closed form (cf.~Zhao et al.~\cite{DBLP:conf/icdcs/ZhaoTLWLX20})
& RB & Delivery of messages within $\delay_r$ rounds; tighten bound from~\cite{DBLP:conf/eurocrypt/PassSS17} using Markov models.\\ 
 	Zhao et al.~\cite{DBLP:conf/icdcs/ZhaoTLWLX20} & $(1-p)^{2 \mu n \maxdelay} > \poweradv$ & RB &  Improved consistency bound from~\cite{DBLP:journals/iacr/KifferRS22} (simpler, closed-form expression).\\
 	Ling Ren~\cite{DBLP:journals/corr/abs-2203-06357} & $\powerhon e^{-2\powerhon\maxdelay} > \poweradv$ & CT & Block generation modeled as a Poisson process.\\ 
 	Dembo et al.~\cite{DBLP:conf/ccs/DemboKTTVWZ20} & $\frac{\powerhon}{1 + \powerhon\lambda\maxdelay} > e \cdot \poweradv$ & CT & Universal condition for all instances of Nakamoto-style blockchains. \\ 
	\bottomrule
\end{tabular}}
\caption{Honest-majority conditions for the security of Nakamoto-style consensus. RB and CT stand for \emph{round-based} and \emph{continuous-time} models respectively. All conditions assume \txtminers have equal power.}
\label{tab:conditions}
\end{table*}

Over the past decade, a series of contributions \cite{DBLP:conf/eurocrypt/GarayKL15, DBLP:journals/iacr/KifferRS22, DBLP:conf/icdcs/ZhaoTLWLX20, DBLP:journals/corr/abs-2203-06357, DBLP:conf/ccs/DemboKTTVWZ20} attempted to shed lights on the necessary and sufficient conditions on the majority assumption under which the security of a blockchain protocol can be achieved. 
Naturally, these conditions improved and got tighter over time. 
Before describing the main results of this work, we start by introducing established security notions for blockchain-based consensus protocols. 

\paragraph{Security Properties} Coining the security properties for blockchains has received considerable attention in the literature. Popular security properties consist of ensuring liveness and persistence.  We shortly summarize these and refer to \cite{DBLP:conf/eurocrypt/GarayKL15, DBLP:conf/ccs/DemboKTTVWZ20} for a formal definition.
\begin{description}
	\item[Persistence] guarantees that if a transaction is included ``deep enough'' in the blockchain of an honest \txtnode (i.e., it is confirmed), eventually it will be included in the local blockchain replicas of all honest \txtnodes, at the same blockchain height. 
	\smallskip\item[Liveness] ensures that every transaction received by honest \txtnodes will eventually be included in the blockchains of all honest \txtnodes as a confirmed transaction (i.e., 
sufficiently deep in the ledger). 
\end{description} 

\noindent \textbf{Existing security models.} We now summarize the main results from prior work. We focus on outlining the various reported conditions under which a blockchain protocol based on Nakamoto-consensus can be proven secure, and summarize the most significant conditions in Table~\ref{tab:conditions}.

Garay \emph{et al.}~\cite{DBLP:conf/eurocrypt/GarayKL15} show that Bitcoin can be proven secure if $\powerhon > \poweradv$. 
Here, \txtminers are assumed to mine in rounds, which we refer to in the following as a round-based model. All \txtminers follow a uniform distribution of hashing power. 
That is, each \txtminer has a fixed number of hash queries per round, i.e., a uniform hash rate. 

Pass \emph{et al.}~\cite{DBLP:conf/eurocrypt/PassSS17} later expanded this round-based model by considering non-zero message delays. 
This improved model uses a network delay bound $\delay_r$, which assumes that all messages are eventually delivered after at most $\delay_r$ rounds. Under those assumptions, the authors proved Bitcoin to be secure when $\powerhon(1-(2\delay_r + 2)\powerhon) > \poweradv$ is satisfied. Kiffer \emph{et al.}~\cite{DBLP:journals/iacr/KifferRS22} extended the security analysis of Pass \emph{et al.} using the Markov model, which resulted in more accurate bounds on consistency. 
Later, Zhao \emph{et al.}~\cite{DBLP:conf/icdcs/ZhaoTLWLX20} further reduced the complexity of the security condition in~\cite{DBLP:journals/iacr/KifferRS22}. 

Later on, Ren~\cite{DBLP:journals/corr/abs-2203-06357} extended these models to capture continuous-time security.
Here, \txtminers can generate and propagate blocks in parallel, where the event of a block being mined is modeled by a Poisson process. 
The results due to Ren~\cite{DBLP:journals/corr/abs-2203-06357} corroborate the results by Zhao et al.
Finally, Dembo et al.~\cite{DBLP:conf/ccs/DemboKTTVWZ20} extended the previous models by making them applicable to all instances of Nakamoto-style blockchains.  
The authors show that both persistence and liveness are guaranteed (up to a value being negligible in some security parameter) if and only if the following inequality holds: 
\begin{equation}\label{eq:CCS-bound_notes}
	e\cdot \poweradv < \frac{1-\poweradv}{1+(1-\poweradv)\rate\maxdelay}.
\end{equation}
Recall that $\poweradv$, $\rate$, and $\maxdelay$ refer to the fraction of adversarial power, the block generation rate, and the maximum communication delay, respectively (cf. \cref{sec:system-model}). 
Moreover, $e$ is a magnification factor that has been specifically introduced to capture the Nothing-at-Stake nature of the Chia Proof-of-Space protocol.

\subsection{Gaps in Existing Models}\label{sec:gaps}
While existing security models significantly boosted the understanding of the security of blockchains, several gaps can be identified. 

First and foremost, existing models provide only limited insight into the security of real-world blockchains. More precisely, apart from \rate, neither \maxdelay nor \poweradv are system parameters that can be tuned directly. Instead, these values depend on other parameters and design choices. So far, these connections have hardly been investigated. 
Second, they indicate that optimal security is given for small networks, e.g., blockchain platforms with $\nodes = 2$ nodes or $\nodes = 3$ nodes, as communication delays will be minimal in this case. 

\paragraph{Motivating Example}
For instance, consider a Nakamoto-style blockchain consisting of three nodes each controlling $\frac{1}{3}$ of the network's power (cf.\ Blockchain A in Figure~\ref{fig:counterexample}). 
Consider another Nakamoto-style blockchain comprising of six nodes each controlling $\frac{1}{6}$ of the total power (cf.\ Blockchain B in Figure~\ref{fig:counterexample}). 
In this example, it is clear that delays witnessed among the nodes in Blockchain B are larger than those witnessed in Blockchain A. Notice, here, that an adversary that controls two out of three nodes in Blockchain A, e.g., nodes $a$ and $b$, has the same relative power as an adversary that controls four out of six nodes, e.g., nodes $\{d, e, f, g\}$, in Blockchain B. 
However, since Blockchain B features a higher overall delay compared to Blockchain A, Blockchain A is deemed to provide better security according to existing models~\cite{DBLP:conf/ccs/DemboKTTVWZ20} (cf.~\cref{eq:CCS-bound_notes}).  
This is clearly not supported in practice: namely, an adversary typically is more likely to succeed in corrupting two nodes in Blockchain A rather than four nodes in Blockchain B. Intuitively, the more nodes a system comprises, the more robust it should be against compromise. Current security models unfortunately do not allow to cover this aspect.

\begin{figure}
	\centering
	\includegraphics[width=\linewidth]{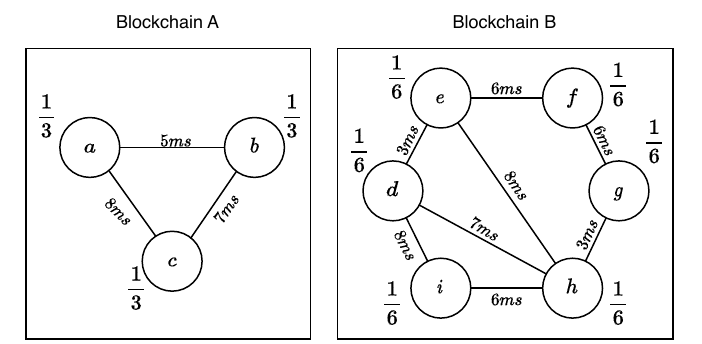}
	\caption{Motivating Example. Existing security models consider Blockchain A (left) to be more secure than Blockchain B (right).} 
	\label{fig:counterexample}
\end{figure}

\paragraph{Research Questions} The overarching goal of our work is 
{to understand the concrete impact of real-world parameters on the overall security of a blockchain platform. 
Towards this goal, we tackle the following research questions:}

\begin{description}[style=unboxed]
\item[RQ1: What is the precise relationship between security and the number of nodes?] The first and most obvious 
parameter that contributes to decentralization is the number of nodes \nodes  (recall that the central motivation behind blockchains is to securely realize functionalities in a decentralized manner). A larger \nodes improves security by making corruption {of a significant fraction of nodes harder, and at the same time it harms} security by causing higher delays. Our work aims to understand the precise effect of these two opposing forces on {security.} 

\smallskip\item[RQ2: Do the security and consistency bounds properly capture issues stemming from practical deployments?] 
{Security analyses of Nakamoto-style consensus protocols assume an upper bound \maxdelay on the communication delays in the network.}
{It remains unclear how these bounds can be established for safe deployments in practice.} Intuitively, 
the larger \nodes, the larger the peer-to-peer network, hence the larger the maximum communication delay \maxdelay. However, besides~\nodes, the underlying choice of the gossip protocol and the block size have a considerable influence on delays; different blockchains feature different gossip protocols and various block sizes. As opposed to \nodes, the choice of the gossip protocol and the block size are under the control of the system designers. Moreover, network attacks, such as~\cite{DBLP:conf/ccs/GervaisRKC15, DBLP:conf/uss/HeilmanKZG15}, have a detrimental impact on the delays witnessed by current blockchains. 
{Another major goal of our work is to assess how and to what extent such attacks affect the security Nakamoto-style consensus in real-world deployments.}

\end{description}

\section{Our Security Framework}\label{sec:decentralized}
In this section, we investigate the impact of decentralization, expressed by the number \nodes of nodes, on the security of Nakamoto-style blockchains. 
Our main result is \cref{th:main} in \cref{sec:OverallSecurity}, that provides a means to assess $\prob{\secevent[\numminers,\numzp]}$, where \secevent is the event that the blockchain is secure (cf. \cref{def:security-event} in \cref{sec:security-model}), \numminers is the number of \txtminers, and \numzp the number of \txtzp. We achieve this as follows:
\begin{enumerate}
	\item In \cref{sec:security-model}, we revisit the results from Dembo et al.~\cite{DBLP:conf/ccs/DemboKTTVWZ20} providing a necessary and sufficient condition for \secevent. This condition depends on two parameters: the maximum delay \maxdelay and the fraction of adversarial power \poweradv. In particular, it allows us to concentrate on $\secevent[\maxdelay,\poweradv]$.
\item To connect $\secevent[\maxdelay,\poweradv]$ to $(\numminers,\numzp)$, we describe afterwards how $(\maxdelay,\poweradv)$ depend on $(\numminers,\numzp)$, i.e., the number of \txtminers and \txtzp, respectively. The impact of \numnodes on the delay, i.e., $\maxdelay(\numnodes)$, is expressed in \cref{th:maxdelay}, which is actually a direct consequence of previous work. The analysis of the relation between the number \numminers of \txtminers and the fraction of adversarial power, i.e., $\poweradv(\numminers)$, is given in \cref{sec:corruption}. Note that for this analysis, we have to extend existing security models and introduce a new corruption model. Namely, existing security models assume an attacker, who is able to \corrupt a constant fraction of \txtminers. \emph{This, however, goes against the very notion of decentralization as it ignores the fact that an increase of \numminers could make the system more robust against \corruption.} Therefore, we instead consider a corruption model, where the success probability of an attacker to corrupt one \txtminer is independent of how many other \txtminers may have been corrupted so far.
\item Finally, in \cref{sec:OverallSecurity} we plug these results together to state in \cref{th:main} how \secevent depends on $(\numminers,\numzp)$. 
\end{enumerate}
Semi-formally, the roadmap in this Section is as follows:
\begin{eqnarray*}
	\secevent&\stackrel{Sec.~\ref{sec:security-model}}{\leftrightarrow}&\secevent[\maxdelay,\poweradv]\\
	&\stackrel{Thm.~\ref{th:maxdelay}\& Sec.~\ref{sec:corruption}}{\leftrightarrow}& \secevent[\maxdelay(\numminers+\numzp),\poweradv(\numminers)]\\
	&\stackrel{Sec.~\ref{sec:OverallSecurity}}{\leftrightarrow}& \secevent[\numminers, \numzp].
\end{eqnarray*}

\subsection{Security Model}\label{sec:security-model}
We start by reformulating the security condition showed by Dembo \emph{et al.}~\cite{DBLP:conf/ccs/DemboKTTVWZ20}. To simplify the discussion, we introduce the following notation to express the security of a Nakamoto-style blockchain:
\begin{definition}[Event of Security $\secevent$]\label{def:security-event}
	For a given blockchain, we denote by \secevent the event that the system provides persistence and liveness. If this does not hold, we write $\neg\secevent$. 
	Moreover, we say that the system is \textit{secure} if and only if \secevent holds.
	
	The term $\prob{\secevent}$ expresses the probability that \secevent is true. We extend this notation as follows: $\prob{\secevent[x_1,\ldots,x_\ell]}$ is likewise the probability that \secevent is true while certain parameters $x_1,\ldots,x_\ell$ of the blockchain are specified.
\end{definition}

Using this notation, we can reformulate the security condition shown in \cite{DBLP:conf/ccs/DemboKTTVWZ20} as follows:
\begin{theorem}\label{th:sec_condition}
Consider a Nakamoto-style blockchain with \numnodes \txtnodes, where all \txtminers possess the same amount of power. The system has a  block frequency \rate and a maximum communication delay of \maxdelay. Moreover, let $e$ be the magnification factor as explained in \cref{sec:hma_tresholds} and \poweradv be the fraction of adversarial power \poweradv.  Then, 
\begin{eqnarray}
	\secevent[\numminers,\numzp] &\Leftrightarrow&e\cdot \poweradv < \frac{1-\poweradv}{1+(1-\poweradv)\rate\maxdelay}\label{eq:conditionDembo}\\
	 &\Leftrightarrow& f(\poweradv)\cdot \rate\cdot \maxdelay < 1, \label{eq:condition}
\end{eqnarray}
where $f(\poweradv)=\left(\frac{e\cdot \poweradv\cdot (1-\poweradv)}{1-\poweradv\cdot (1+e)}\right).$
\end{theorem}
Recall that $\poweradv$ (cf. \cref{eq:rel_power_def}) takes into account the adversary's relative power (i.e., computing power or amount of stake) when generating new blocks, e.g., the mining difficulty for PoW blockchains.
Equation~(\ref{eq:condition}) in \cref{th:sec_condition} directly follows by rearranging Equation~(\ref{eq:conditionDembo}).
In this equation, we have three factors $f(\poweradv)$, $\rate$, and $\maxdelay$, and security is only guaranteed if and only if their product is below 1. Before we analyze the relation into more details, we first discuss the implications of the first factor, $f(\poweradv)$. 

\begin{lemma}\label{lem:function}
	The function $f(\poweradv)$ is strictly increasing in $\poweradv \in [0,\frac{1}{1+e})$. Moreover, it holds that $f(0)=0$ and $\lim_{\poweradv\rightarrow \frac{1}{1+e}}f(\poweradv)=\infty$.
\end{lemma}
\begin{proof}
	Consider the derivative of $f(\poweradv)$:
	\begin{equation*}
		f'(\poweradv)=\frac{e\cdot \left( (e+1)\cdot \poweradv^2-2\poweradv+1\right)}{\left((e+1)\cdot \poweradv-1\right)^2}.
	\end{equation*}
	Since  $e\geq 1$, it holds that 
	\begin{eqnarray*}
		e\cdot \left( (e+1)\cdot \poweradv^2-2\poweradv+1\right)&>&\left(\poweradv^2-2\poweradv+1\right)\\
		&=&\left(\poweradv-1\right)^2\\
		&>&0	
	\end{eqnarray*}
	and $\left((e+1)\cdot \poweradv-1\right)^2>0$ as well. \
	Thus, $f'(\poweradv) > 0$ in the considered domain. 
	The second part of Lemma~\ref{lem:function} follows easily by inspecting $f(x)$.
\end{proof}

In consequence, we see a tradeoff between security ($f(\poweradv)$) and efficiency ($\rate\cdot \maxdelay$). Namely, an increase of $f(\poweradv)$ (which is equivalent to an increase of \poweradv as $f$ is strictly monotonically rising, cf. \cref{lem:function}) does not break security as long as $\maxdelay$ is sufficiently small and vice versa. In the sequel, we investigate how these two parameters depend on the number \numminers of \txtminers and \numzp of \txtzp.

With respect to the impact of the number of validators on \maxdelay, we can actually make use of the following known result based on ~{\cite[Theorem 6]{CHUNG2001257}} (see Appendix~\ref{ap:thmaxdelay} for more details).
\begin{theorem}[Maximum Delay]\label{th:maxdelay}
	For a blockchain with \numnodes \txtnodes, it holds that: 
	\begin{equation}
		\maxdelay(\nodes) \in\Theta(\log(\nodes)). \label{eq:maxdelay}
	\end{equation}
\end{theorem}A consequence is that $\maxdelay(\nodes)$ can be approximated by $a\cdot \log(\nodes)+b$ for some constants $a,b$.
In particular, $\maxdelay(\nodes)$ increases with if the number \nodes increases.

\subsection{Corruption Model} \label{sec:corruption}
We now investigate how the second parameter from \cref{th:sec_condition}, namely the relative adversarial power \poweradv, depends on the number \numnodes of nodes. More precisely, we consider the number \numminers of \txtminers only, as \txtzp do not possess any mining power and hence no \corrupted \txtzp does  influence \poweradv. 
In what follows, we first explain the considered \corruption model and derive bounds on \poweradv afterwards in \cref{th:bounding-poweradv}.

In general, an attacker can \corrupt \txtminers to control their resources (cf. \cref{sec:system-model}). 
Contrary to previous models \cite{DBLP:conf/eurocrypt/GarayKL15, DBLP:conf/eurocrypt/PassSS17, DBLP:journals/iacr/KifferRS22, DBLP:conf/icdcs/ZhaoTLWLX20, DBLP:journals/corr/abs-2203-06357, DBLP:conf/ccs/DemboKTTVWZ20}, we do not impose an upper bound on the number of \corrupted nodes. Instead, 
we assume that each \txtminer can be probabilistically \corrupted independently of the other \txtminers, in line with Gentry \etal~\cite{DBLP:conf/crypto/GentryHKMNRY21}. Moreover, to ensure that our model generalized to realistic scenarios, we also assume that the corruption probability can differ among different validators. That is, our proposed corruption model allows to capture situations where the validators with a huge amount of resources are better protected and hence more resilient against corruption and vice versa. This generalization, however, imposes another question: as we consider an increasing number of validators, we need to decide the corruption probability for each validator that newly joins the blockchain. Our approach here is to divide the set of validators into a number of different types, let's say $\ell$. Each type $i$ has a different probability $\crprobunf^{(i)}$ of being corrupted. Moreover, there are probabilities $\charprobunf^{(1)},\ldots,\charprobunf^{(\ell)}$ where $\charprobunf^{(i)}$ denotes the probability that a freshly joining validator is of type $i$, that is can be corrupted with probability $\crprobunf^{(i)}$. 
To formally capture this,  we introduce the notion of a characterization:
\begin{definition}[Characterization]\label{def:characterization}
A characterization \charac is a sequence of value pairs
\begin{equation}
	\charac=\left((\crprobunf^{(1)},\charprobunf^{(1)}),\ldots, (\crprobunf^{(\ell)},\charprobunf^{(\ell)})\right) 
\end{equation}
such that $\crprobunf^{(i)},\charprobunf^{(i)}\in[0,1]$ for $i=1,\ldots,\ell$
and $\sum_{i=1}^{\ell}\charprobunf^{(i)}=1$.
\end{definition}

We are now ready to define a random variable \RVreladv that expresses the fraction of adversarial power. {Recall that \numminers denotes the number of validators in the network.} 
\begin{definition}[Random Variable $\RVreladv(\numminers,\charac)$]\label{def:RV}
	We consider the following random experiment, based on an integer \numminers and a characterization~$\charac=\left((\crprobunf^{(1)},\charprobunf^{(1)}),\ldots, (\crprobunf^{(\ell)},\charprobunf^{(\ell)})\right)$. 
In the experiment, the following procedure is repeated for $j=1,\ldots,\numminers$. First, a validator type is sampled according to the distribution $\left(\charprobunf^{(1)},\ldots,\charprobunf^{(\ell)}\right)$. Let the outcome be denoted by $i_j$. Then, with a probability of $\crprobunf^{(i_j)}$, the validator is marked as corrupted.
	
	The random variable $\RVreladv(\numminers,\charac)$ is defined by the fraction of successful \corruption attempts, that is 
	\begin{equation}
		\RVreladv(\numminers,\charac)=\frac{\#\ \mathrm{\corrupted\ \txtminers}}{\numminers}.
	\end{equation}
	\end{definition}

It holds that the probability of a new validator of getting corrupted is
\begin{equation}
\EV:=\sum_{i=1}^{\ell}\charprobunf^{(i)}\cdot \crprobunf^{(i)}
\end{equation}
Note that value of $\RVreladv(\numminers,\charac)$ is exactly the parameter \poweradv from \cref{th:sec_condition}. Our approach allows to express the probability that \poweradv is below a certain threshold. Namely, 
$\prob{\RVreladv(\numminers)\leq \frac{k}{\numminers}}$ for any choice of $k$ is equal to the probability that up to $k$ out of \numminers \corruption attempts have been successful.
It therefore holds that:
\begin{eqnarray}
	&&\prob{\RVreladv(\numminers)\leq \frac{k}{\numminers}}\nonumber\\
& =& \sum_{i=0}^{k } \binom{{\numminers}}{i} \left(p^*\right)^{i} \cdot  \left(1-p^*\right)^{\numminers-i} \label{eq:probRVBound}
\end{eqnarray}

For analyzing the asymptotic behavior of $\RVreladv(\numminers,\charac)$ (in dependence of $\numminers$), we will derive some bounds for $\RVreladv(\numminers,\charac)$. 
\begin{theorem}[Bounding $\RVreladv(\numminers,\charac)$]\label{th:bounding-poweradv}
	
	Let $\charac=\left( (\crprobunf^{(i)},\charprobunf^{(i)})\right)_{i=1}^\ell$ and $\EV:=\sum_{i=1}^{\ell}\charprobunf^{(i)}\cdot \crprobunf^{(i)} $. Then, it holds for any
	$0 \leq  \varepsilon < p^*$ that 
		\begin{equation}
		\prob{\left|\RVreladv(\numminers,\charac)-\EV\right|\geq \varepsilon}\leq 2e^{-\numminers \cdot \frac{\varepsilon^2}{3\cdot \EV}}.	\label{eq:boundsWithEpsilon}
	\end{equation}

\end{theorem}

An immediate consequence of \cref{th:bounding-poweradv} is the following:
\begin{corollary}[Limits of $\RVreladv(\numminers,\charac)$]\label{cor:limes}
	Let~$\charac=\left( (\crprobunf^{(i)},\charprobunf^{(i)})\right)_{i=1}^\ell$. It holds that
	\begin{equation}
		\lim_{\numminers\rightarrow \infty}\RVreladv(\numminers,\charac)=\sum_{i=1}^{\ell}\charprobunf^{(i)}\cdot \crprobunf^{(i)}. \label{eq:RV_approximation}
	\end{equation}
\end{corollary}

To prove \cref{th:bounding-poweradv}, we make use of the following Chernhoff Bound:
\begin{lemma}[Chernhoff Bound]
	\label{lemma:chernhoff}
	Let $X_1,\ldots ,X_n$ be independent random variables such that $X_i$ always lies in the interval $[0; 1]$. Define $X=\sum_i X_i$ and $\mu=\ExpVal{}{X}$. 
	Then for any $d \in (0,1)$, it holds that 
	\begin{equation}
		\prob{|X-\mu|\geq d\cdot \mu}\leq 2e^{-\frac{\mu\cdot d^2}{3}}.\label{eq:twosidebound}
	\end{equation}
\end{lemma} 
We are now ready to prove \cref{th:bounding-poweradv}.
\begin{proof}[Proof of \cref{th:bounding-poweradv}]
	We consider the random experiment as defined in \cref{def:RV}. For $i=1,\ldots,\numminers$, we define 
	the following Bernoulli random variables:
	\begin{equation}
		\RVabs_i = 
		\begin{cases}
			1, & \text{ if $i$-th \txtminer is corrupted}\\
			0, & \text{otherwise.}
		\end{cases}
	\end{equation}
	Recall that for any \txtminer, 
	$\charprobunf^{(j)}$ expresses the probability that it can be corrupted with probability $\crprobunf^{(j)}$.
	As a consequence, we can rephrase the random variables $\RVabs_i$ as follows:
	\begin{equation}
		\RVabs_i = 
		\begin{cases}
			1, & \text{with probability $\EV$} \\
			0, & \text{else}
		\end{cases}
	\end{equation}
	with $\EV=\sum_{j=1}^{\ell}\charprobunf^{(j)}\cdot \crprobunf^{(j)}$. Observe that 
	$X=\sum_{i\in[\numminers]}\RVabs_i$ is the number of \corrupted \txtminers and hence 
	\begin{equation}
		\RVreladv(\numminers,\charac)=\frac{X}{\numminers}=\frac{\sum_{i\in[\numminers]}\RVabs_i}{\numminers}.\label{eq:RelationRVandX}
	\end{equation}
	Moreover, the expected value of a single random variable $\RVabs_i$ is
	\begin{equation}
		\ExpVal{}{\RVabs_i}= \EV.
	\end{equation}
	As the variables $\RVabs_i$ are pairwise independent, it follows that 
	\begin{equation}
		\ExpVal{}{X} =  \ExpVal{}{\sum_{i=1}^{\numminers}\RVabs_i} = \numminers\cdot \EV.
	\end{equation}
	The Chernoff bound (\cref{lemma:chernhoff}) gives for any $d\in(0,1)$
	\begin{eqnarray*}
		\prob{\left|X - \numminers  \cdot \EV \right| \geq  d \cdot (\numminers  \cdot \EV)} &\leq& 2e^{-\frac{\numminers\cdot \EV d^2}{3}}\\
		\Leftrightarrow \prob{\left|\frac{X}{\numminers}-\EV\right|\geq  d\cdot \EV } &\leq& 2e^{-\frac{\numminers\cdot \EV d^2}{3}}\\
		\stackrel{\eqref{eq:RelationRVandX}}{\Leftrightarrow} \prob{\left|\RVreladv(\numminers,\charac)-\EV\right|\geq  d\cdot \EV } &\leq& 2e^{-\frac{\numminers\cdot \EV d^2}{3}}\\
	\end{eqnarray*}
	If we express $d:=\frac{\varepsilon}{\EV}$ for a sufficiently small $\varepsilon$ such that the condition $0<\delta<1$ is preserved, i.e., $0 \leq  \varepsilon < \EV$, then $d\cdot \EV$ can be replaced by $\varepsilon$ and we get \cref{eq:boundsWithEpsilon}.
\end{proof}

\subsection{Security Analysis}\label{sec:OverallSecurity}
Next, we use the previous results to derive our main theorem:
\begin{theorem}\label{th:main}
	Consider a Nakamoto-style blockchain with $\numnodes=\numminers+\numzp$ \txtnodes with \numminers being the number of \txtminers and \numzp being the number of \txtzp. 
	Let $p^*$ denote the probability for an arbitrary \txtminer to be \corrupted. 
	
	It holds that:
	\begin{multline} 
		\prob{\secevent[\numminers,\numzp]} = \\ \sum_{i=0}^{\overbrace{\lfloor \poweradvmax\cdot \numminers \rfloor}^{:=g(n)}} \binom{{\numminers}}{i} \left(p^*\right)^{i} \cdot \left(1-p^*\right)^{\numminers - i}  \label{eq:success} 
	\end{multline}
	with 
	\begin{equation} \label{eq:adv_bound}
	\poweradvmax = \frac{\varMainTheorem + 1 + e}{2\varMainTheorem} - \sqrt{\left(\frac{\varMainTheorem + 1 + e}{2\varMainTheorem}\right)^2 - \frac{1}{\varMainTheorem}}, 
	\end{equation} 
	where $\varMainTheorem \in \Theta\left(e\rate\log(\numminers+\numzp)\right)$.
\end{theorem}
Note that function $g(\nodes)$ captures what is commonly referred to as the ``Nakamoto coefficient'' and refers to the minimum number of nodes that have to be corrupted in order for an adversary to successfully attack a Nakamoto-style blockchain.

\begin{proof}[Proof of \cref{th:main}]
Because of \cref{th:sec_condition}, it holds that
\begin{eqnarray}
	\secevent[\numminers,\numzp]
	 &\Leftrightarrow&e\cdot \poweradv < \frac{1-\poweradv}{1+(1-\poweradv)\rate\maxdelay}\\
	&\Leftrightarrow&\poweradv\leq  \poweradvmax(\maxdelay).
\end{eqnarray}
Thus, it holds that:
\begin{equation}
	\prob{\secevent[\numminers,\numzp]} = \prob{\RVreladv(\numminers)\leq \poweradvmax}
\end{equation}
\cref{eq:success} is a direct consequence of \cref{eq:probRVBound}.

Moreover, it holds that
\begin{eqnarray}
	\poweradvmax(\maxdelay) :=& \frac{e\rate\maxdelay + 1 + e}{2e\rate\maxdelay} - \sqrt{\left(\frac{e\rate\maxdelay + 1 + e}{2e\rate\maxdelay}\right)^2 - \frac{1}{e\rate\maxdelay}} \label{eq:max_adv_power} \\
	=&\frac{\varMainTheorem + 1 + e}{2\varMainTheorem} - \sqrt{\left(\frac{\varMainTheorem + 1 + e}{2\varMainTheorem}\right)^2 - \frac{1}{\varMainTheorem}},
\end{eqnarray}
where $\varMainTheorem = e\rate\maxdelay$.
Since $\maxdelay\in \Theta\left(\log(\numminers+\numzp)\right)$ (cf. \cref{th:maxdelay}), we conclude that $\varMainTheorem \in \Theta\left(e\rate\log(\numminers+\numzp)\right)$.
\end{proof}

In \cref{fig:SG_Rate}, we quantify the security probability according to \cref{eq:success} with respect to the block generation rate $\rate$. We observe that $\prob{\secevent[\numminers,\numzp]}$ is almost 100\% for modest block generation rates $\rate \leq 10^{-4}$ and tends to 0\% for $\rate \geq 10^{-1}$. We also observe that $\prob{\secevent[\numminers,\numzp]}$ decreases as $\EV$ and \nodes increase (cf.~\cref{sec:network_attack}). 

Further, \cref{th:main} allows us to express the impact of \nodes on the overall security of a blockchain. Actually, it holds that if \nodes keeps on increasing, the blockchain inevitably becomes insecure:
\begin{lemma}[Limits of $\prob{\secevent[\numminers,\numzp]}$ ]\label{lem:limesSecurity}
	It holds that
	\begin{equation}
		\lim_{\numnodes=\numminers+\numzp\rightarrow\infty}\prob{\secevent[\numminers,\numzp]}=0.
	\end{equation}
\end{lemma}

\begin{proof}[Proof]
	Due to \cref{th:sec_condition}, it holds that: 
	\begin{equation}
		\secevent[\numminers,\numzp]\Leftrightarrow f(\poweradv(\numminers))\cdot \rate\cdot \maxdelay(\numminers,\numzp) < 1.\label{eq:limesConditions}
	\end{equation}
	To show the claim for $\lim_{\numnodes=\numminers+\numzp\rightarrow\infty}$, we need to prove the following 2 statements: 
	\begin{eqnarray*}
		\lim_{\numzp\rightarrow\infty}\prob{\secevent[\numminers,\numzp]}&=&0\ \mathrm{and}\\
		\lim_{\numminers\rightarrow\infty}\prob{\secevent[\numminers,\numzp]}&=&0.
	\end{eqnarray*}
	Let $\numminers$ be some fixed value. Recall that $\maxdelay(\numminers,\numzp)\in\mathcal{O}(\log(\numzp))$ due to \cref{th:maxdelay}. Hence, we have
	\begin{equation}
		\lim_{\numzp\rightarrow\infty}(\underbrace{f(\poweradv(\numminers))}_{=\mathrm{const.}}\cdot \rate\cdot \underbrace{\maxdelay(\numminers,\numzp)}_{\rightarrow \infty})=\infty.
	\end{equation}
	
	In particular, it is guaranteed that for a sufficiently large \numzp, the condition given in \cref{eq:limesConditions} cannot be met. Moreover:
	\begin{equation*}
		\lim_{\numzp\rightarrow\infty}\prob{\secevent[\numminers,\numzp]}=0.
	\end{equation*}
	This proves the first statement.
	
	To prove the second statement, we make use of the fact that due to \cref{cor:limes}, it holds that 
	$\lim_{\numminers\rightarrow\infty}\poweradv(\numminers)=\mathrm{const.}$. Thus, similar to above, it follows
	\begin{equation}
		\lim_{\numminers\rightarrow\infty}(\underbrace{f(\poweradv(\numminers))}_{\rightarrow \mathrm{const.}}\cdot \rate\cdot \underbrace{\maxdelay(\numminers,\numzp)}_{\rightarrow \infty})=\infty.
	\end{equation}
\end{proof}

\noindent \textbf{Summary of the results. }
\cref{th:sec_condition} 
illustrates the relation between the three parameters \poweradv, \rate, and \maxdelay. In particular, the higher one of these values, the lower the others need to be. 
In combination with \cref{th:maxdelay}, it confirms the common belief that an increase in the number of nodes increases the communication delay and  hence reduces the set of possible safe parameters. However, it also shows that the impact on \maxdelay is only logarithmic, which is confirmed by our experiments in \cref{sec:simulator}.

This is also important when one relates the impact of \nodes on \maxdelay with its impact on \poweradv. On the one hand, \cref{lem:limesSecurity} shows that if a blockchain gets too large (in terms of \nodes), guaranteeing security is impossible. Thus, the decentralization approach has its limits and having too many \txtminers can be counterproductive. 
On the other hand, an increase of \numminers makes the system more robust against \corruption.  From \cref{th:bounding-poweradv}, we know that $\prob{\RVreladv(\numminers)\leq \frac{k}{\numminers}}$ for some fixed $k$ is subject to high variations for small \numminers but becomes more stable with increasing \numminers.  

That is, we have two concurrent trends\footnote{As an increase of \numzp only increases the communication delay without helping to protect against corruption, we ignore \numzp  from now on and focus on $\prob{\secevent[\numminers]}$. }: increasing \numminers negatively impacts security due to the increase of \maxdelay, and at the same time, positively impacts security as high values of \poweradv become less likely. 

In Section~\ref{sec:real_world_analysis}, we investigate the impact of real-world parameters on these tradeoffs.

\begin{figure}
	\centering
	\includegraphics[width=1.0\linewidth]{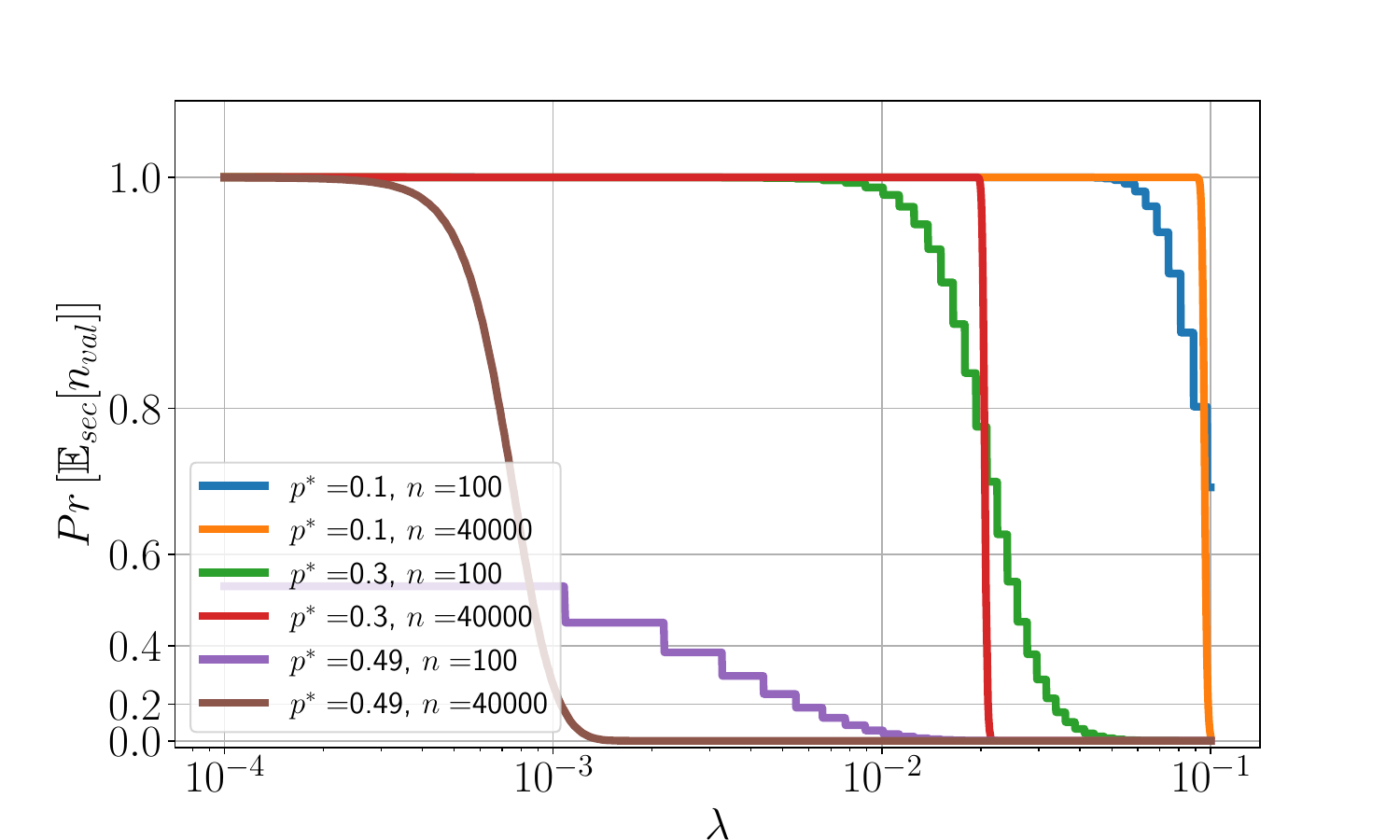}
	\caption{Impact of the block generation rate on the security probability  (cf. \cref{eq:success}). } 
	\label{fig:SG_Rate}
\end{figure}

\section{Empirical Validation} \label{sec:simulator}

To analyze the outcomes of our analysis (cf.~\cref{sec:decentralized}) on current real-world deployments, we conducted large-scale measurements to derive realistic upper bounds on the network delay as the network scale grows. Our ultimate goal is to establish a precise relation between the number of nodes~\nodes and the maximum network delay~\maxdelay for each of the analyzed Nakamoto-style deployments. Besides validating the theoretical results, this first set of experiments allows us to derive protocol-specific constants~$a,b$ such that $\maxdelay(n) \approx a \cdot \log(n)+b$ (cf.~\cref{th:maxdelay}).

\subsection{Methodology} \label{sec:methodology}
\newcommand{\ncp}{\ensuremath{p^{nt}_{adv}}\xspace}
\newcommand{\ncs}{\ensuremath{p^{nt}_{con}}\xspace}
\newcommand{\nct}{\ensuremath{nt_{delay}}\xspace}

\newcommand{\push}{\textit{direct push}\xspace}
\newcommand{\adv}{\textit{advertisement-based}\xspace}
\newcommand{\hybrid}{\textit{hybrid push}\xspace}
\newcommand{\cbr}{\textit{compact blocks}\xspace}
\newcommand{\inv}{\textit{inv}\xspace}
\newcommand{\txtzps}{zero-power nodes\xspace}
\newcommand{\maxdelayq}{\ensuremath{\Delta_{0.9}}\xspace}  

We opted to validate our findings empirically by means of simulations. Indeed, our analysis explores the security of large-scale blockchains, i.e., those comprising of hundreds of thousands of nodes. Simulations emerge as one of the few workable options to empirically evaluate such large-scale networks (it is virtually impossible to deploy a blockchain network with hundreds of thousands of nodes otherwise).

All our experiments were executed on a server running Ubuntu 18.04 and OpenJDK 11.07. The server is equipped with two AMD EPYC 7642 48-Core Processors and 256 GB of RAM. Unless stated otherwise, every data point in our plots is averaged over 5 independent runs of 100 blocks. 

\paragraph{Overview of SimBlock} 
Our implementation is based on the open-source SimBlock simulator~\cite{DBLP:conf/infocom/AokiOKBS19} implemented in Java.
SimBlock is a simulator designed specifically for evaluating the block propagation in real-world blockchain networks.
Its primary purpose is to provide a controlled environment for assessing blockchain performance and security. 
Towards this end, it simulates a network of nodes placed in different regions according to some input geographical distribution, 
and derives the network performance (i.e., throughput and latency) of each node from a configuration file specifying delays between regions (i.e., a node in Europe has a lower latency to send to another node in Europe compared to Asia, North America, etc.). The total delay to send a message from node A to node B is then based on the message size, the minimum between the upload bandwidth of node A and the download bandwidth of node B, and the maximum between the two nodes' latencies.

SimBlock simulates communication between blockchain nodes based on events. 
Events are triggered by nodes sending messages over the network and generating blocks.
SimBlock discretizes time using milliseconds as the base unit and constructs a global timeline of all events based on their chronological order so that they are processed sequentially (simultaneous events are processed in an arbitrary order).
Events associated to a message are added to the global timeline based on their reception time.
For each message, the reception time is interpolated from the propagation delay (defined as the maximum latency between the sender and the receiver) and the transmission delay (defined as the size of the message divided by the minimum throughput between the sender and receiver). 

The simulation is initiated by a randomly selected node generating the genesis block and adopting it as the longest chain.
Upon reception or adoption of a new block, nodes proceed with the propagation by compiling a list of messages to transmit to their neighboring nodes, based on the specified propagation mechanism (cf. Figure~\ref{fig:gossip}).
If the node is a non-zero miner, the node further randomly samples the time to generate the next block following an exponential distribution 
with $\lambda=\frac{1}{600}*\frac{\absmp}{\absmp(\setminers)}$, where $\frac{\absmp}{\absmp(\setminers)}$ represents the node's relative mining power; the resulting block generation time is then added to the timeline.

Events within SimBlock encompass all message exchanges between nodes (e.g., \inv messages in Figure~\ref{fig:gossip}) as well as block generation events.
This allows us to accurately measure the propagation of different messages and blocks, as well as implementing different adversarial strategies, such as incurring artificial delay on certain links. Notice that simulations are deterministic assuming the same initial random seed. More importantly, the output of the simulator does not depend on hardware configurations, such as the number of cores or the amount of RAM.

\paragraph{Calibration}  To ensure a realistic setup,
we calibrated the configuration file specifying the network's bandwidth and latency so as 
to yield comparable block-propagation times and stale-block rate to the KASTEL dataset~\cite{DBLP:conf/p2p/DeckerW13, kastel}. We also relied on the geographical distribution of nodes as reported by~\cite{bitnodes}.
We include additional details on the parameter calibration of the simulator and the resulting average throughput and latency for each region in Appendix~\ref{ap:calibration}.

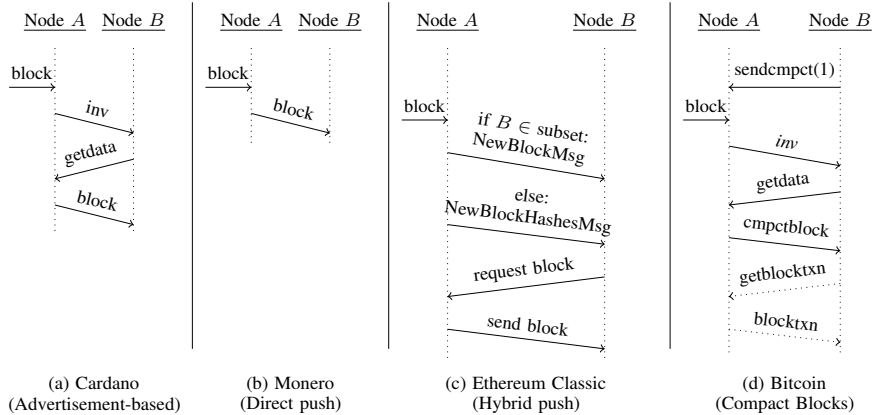
\begin{figure*}[h!]
	\centering
	\scalebox{0.87}{
	\begin{tikzpicture}
		\tikzstyle{every node}=[font=\footnotesize]
		\def\xab{0}
		\def\xhp{6}
		\def\xdp{3}
		\def\xcmpt{10.3}

		\def\abwidth{1.2}
		\node at (\xab + \abwidth/2, -5.7) [text width=2.8cm, align=center] (label_ab) {(a) Cardano\\(Advertisement-based)};

		\node at (0+\xab, 0) (A1) {\uline{$\text{Node}\ A$}};
		\node at (\abwidth+\xab, 0) (B1) {\uline{$\text{Node}\ B$}};
		\draw[dotted] (0+\xab, -0.4) -- (0+\xab, -3.2);
		\draw[dotted] (\abwidth+\xab, -0.4) -- (\abwidth+\xab, -3.2);

		\draw[->] (-0.7+\xab, -1) -- (0+\xab, -1) node [midway, above, sloped] (TextNode) {block};

		\draw[->] (0+\xab, -1.4) -- (\abwidth+\xab, -1.7) node [midway, above, sloped] (TextNode) {inv};
		\draw[->] (\abwidth+\xab, -2.1) -- (0+\xab, -2.4) node [midway, above, sloped] (TextNode) {getdata};
		\draw[->] (0+\xab, -2.8) -- (\abwidth+\xab, -3.1) node [midway, above, sloped] (TextNode) {block};

		\def\xa{0}
		\def\xb{2.4}
		\def\xc{0.5}
		\def\xd{0.3}

		\node at (\xhp + \xb/2, -5.7) [text width=4.0cm, align=center] (label_ab) {(c) Ethereum Classic\\(Hybrid push)};

		\node at (\xa+\xhp, 0) (A2) {\uline{$\text{Node}\ A$}};
		\node at (\xb+\xhp, 0) (B2) {\uline{$\text{Node}\ B$}};
		\draw[dotted] (\xa+\xhp, -0.4) -- (\xa+\xhp, -5.2);
		\draw[dotted] (\xb+\xhp, -0.4) -- (\xb+\xhp, -5.2);

		\draw[->] (\xa-0.7+\xhp, -1.5) -- (\xa+\xhp, -1.5) node [midway, above, sloped] (TextNode) {block};

		\draw[->] (\xa+\xhp, -2.0) -- (\xb+\xhp, -2.0-1*\xd-1*0.1) node [midway, above, sloped,text width=2.0cm,align=center] (TextNode) {if $B \in$ subset: \\ NewBlockMsg};
		\draw[->] (\xa+\xhp, -1.5-2*\xd-2*\xc) -- (\xb+\xhp, -1.5-3*\xd-2*\xc) node [midway, above, sloped,text width=3.0cm,align=center] (TextNode) {else:\\NewBlockHashesMsg};
		\draw[->] (\xb+\xhp, -1.5-3*\xd-3*\xc) -- (\xa+\xhp, -1.5-4*\xd-3*\xc) node [midway, above, sloped] (TextNode) {request block};
		\draw[->] (\xa+\xhp, -1.5-4*\xd-4*\xc) -- (\xb+\xhp, -1.5-5*\xd-4*\xc) node [midway, above, sloped] (TextNode) {send block};

		\def\dpwidth{1.2}
		\node at (\xdp + \dpwidth/2, -5.7) [text width=4.0cm, align=center] (label_ab) {(b) Monero\\(Direct push)};

		\node at (0+\xdp, 0) (A1) {\uline{$\text{Node}\ A$}};
		\node at (\dpwidth+\xdp, 0) (B1) {\uline{$\text{Node}\ B$}};
		\draw[dotted] (0+\xdp, -0.4) -- (0+\xdp, -1.9);
		\draw[dotted] (\dpwidth+\xdp, -0.4) -- (\dpwidth+\xdp, -1.9);

		\draw[->] (-0.7+\xdp, -1) -- (0+\xdp, -1) node [midway, above, sloped] (TextNode) {block};

		\draw[->] (0+\xdp, -1.4) -- (\dpwidth+\xdp, -1.7) node [midway, above, sloped] (TextNode) {block};

		\def\xa{0}
		\def\xb{1.7}
		\def\xc{0.5}
		\def\xd{0.2}
		\def\ya{3.5}
		\def\yb{5.2}
		\def\yc{0.5}
		\def\yd{0.2}

		\node at (\xcmpt + \xb/2, -5.7) [text width=4.0cm, align=center] (label_ab) {(d) Bitcoin \\ (Compact Blocks)}; 

		\node at (\xa+\xcmpt, 0) (A) {\uline{$\text{Node}\ A$}};
		\node at (\xb+\xcmpt, 0) (B) {\uline{$\text{Node}\ B$}};
		\draw[dotted] (\xa+\xcmpt, -0.4) -- (\xa+\xcmpt, -5.1);
		\draw[dotted] (\xb+\xcmpt, -0.4) -- (\xb+\xcmpt, -5.1);

		\draw[->] (\xb+\xcmpt, -1) -- (\xa+\xcmpt, -1) node [midway, above, sloped] (TextNode) {sendcmpct(1)};

		\draw[->] (\xa-0.7+\xcmpt, -1.5) -- (\xa+\xcmpt, -1.5) node [midway, above, sloped] (TextNode) {block};

		\draw[->] (\xa+\xcmpt, -1.9) -- (\xb+\xcmpt, -1.9-1*\xd-1*0.1) node [midway, above, sloped] (TextNode) {\textit{inv}};
		\draw[->] (\xb+\xcmpt, -1.9-1*\xd-1*\xc) -- (\xa+\xcmpt, -1.9-2*\xd-1*\xc) node [midway, above, sloped] (TextNode) {getdata};
		\draw[->] (\xa+\xcmpt, -1.9-2*\xd-2*\xc) -- (\xb+\xcmpt, -1.9-3*\xd-2*\xc) node [midway, above, sloped] (TextNode) {cmpctblock};
		\draw[->, dotted] (\xb+\xcmpt, -1.9-3*\xd-3*\xc) -- (\xa+\xcmpt, -1.9-4*\xd-3*\xc) node [midway, above, sloped] (TextNode) {getblocktxn};
		\draw[->, dotted] (\xa+\xcmpt, -1.9-4*\xd-4*\xc) -- (\xb+\xcmpt, -1.9-5*\xd-4*\xc) node [midway, above, sloped] (TextNode) {blocktxn};

		\draw[-] (2.1, 0.3) -- (2.1, -5.0);
		\draw[-] (5.1, 0.3) -- (5.1, -5.0);
		\draw[-] (9.4, 0.3) -- (9.4, -5.0);
	\end{tikzpicture}}
	\caption{Summary of the popular gossip protocols used in existing Nakamoto-style blockchains.}
	\label{fig:gossip}
\end{figure*}

\paragraph{Modifications to SimBlock} 
Unfortunately,
the publicly available version of SimBlock exhibits serious scalability limitations.
For instance, the maximum scale of a blockchain network that we could emulate was initially limited to around \num{70000} \txtminers in our experimental setup due to limitations in the precision of floating points when calculating the difficulty of the blocks.
Therefore, we significantly extended SimBlock as follows:
\begin{description}

\item \textbf{Fine-grained difficulty }Since Java's float data type is limited to 64 bits, the accuracy of the difficulty is limited when dealing with a big value of total block generation power.
We bypassed this limitation by encoding only the relative block generation power of each node instead of the total block generation power.
\smallskip
\item[Gossip protocol variants] To account for various deployments, beyond Bitcoin and altcoins, such as Monero and Cardano, we additionally implemented the hybrid and direct push protocols and revised the low-bandwidth implementation of the compact block relay protocols to incorporate the transmissions of the associated transactions (see Figure~\ref{fig:gossip}). Namely, SimBlock did not originally account for the delays introduced for the transmission of transactions ---
these were shown to have a significant impact on the performance of the network~\cite{DBLP:conf/ndss/0003ZWWXP22}.
\smallskip
\item[Network-layer attacks] We further extended SimBlock to capture a collection of realistic attacks stemming from the network layer.
For this purpose, we introduce a number of adversarial nodes in our network that could delay their block-propagation messages.
\smallskip

\item[Memory and event timeline optimization] In order to handle tasks that may need to be removed from the timeline, the original code kept a map between tasks and their respective events in the timeline. When a task had to be removed, the task's event would be retrieved from the map before being removed from both the timeline and the map. In our simulator, we reduced the scope of this map to only tasks that can be removed, i.e., block generation tasks, effectively reducing its size by more than 95\%. We also optimized the implementation of the priority queue to only include events that are unlikely to be immediately removed from the queue.

\end{description}

\begin{figure*}
	\begin{subfigure}[t]{0.326\textwidth}
		\includegraphics[width=\linewidth]{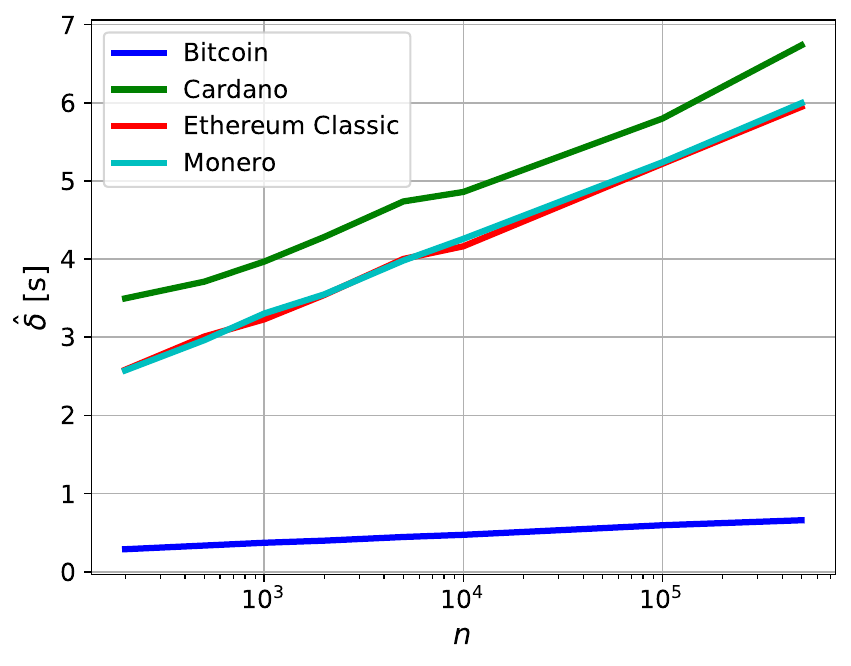}
		\caption{\avgdelay against \nodes.}
		\label{fig:uniform_mean_benign}
	\end{subfigure}
	\begin{subfigure}[t]{0.33\textwidth}
		\includegraphics[width=\linewidth]{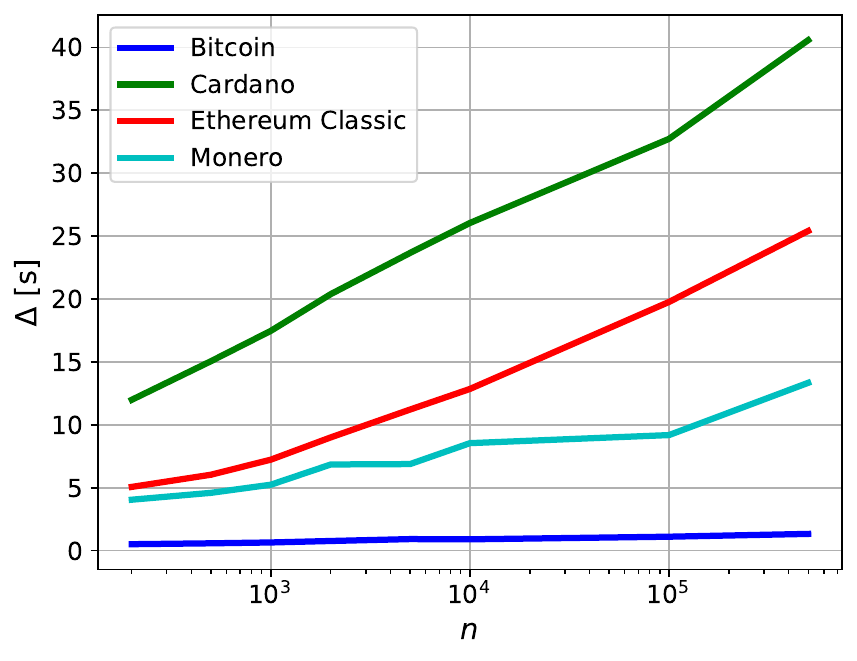}
		\caption{\maxdelay against \nodes.}
		\label{fig:uniform_max_benign}
	\end{subfigure}
	\begin{subfigure}[t]{0.32\textwidth}
			\vspace{-4.5cm}
	\scalebox{0.75}{{\color{red}}
			\begin{tabular}{cl}
				\toprule
				Cryptocurrency & Function\\
				\midrule
				\multirow{2}*{Bitcoin}&$\maxdelay(\nodes)=\textbf{0.1}\cdot log(\nodes)-0.04$\\
				&$\avgdelay(\nodes)=\textbf{0.04}\cdot log(\nodes)+0.03$\\
				\midrule
				\multirow{2}*{Cardano}&$\maxdelay(\nodes)=\textbf{3.65}\cdot log(\nodes)-7.37$\\
				&$\avgdelay(\nodes)=\textbf{0.41}\cdot log(\nodes)+1.3$\\
				\midrule
				\multirow{2}*{Ethereum Classic}&$\maxdelay(\nodes)=\textbf{2.6}\cdot log(\nodes)-8.71$\\
				&$\avgdelay(\nodes)=\textbf{0.42}\cdot log(\nodes)+0.3$\\
				\midrule
				\multirow{2}*{Monero}&$\maxdelay(\nodes)=\textbf{1.18}\cdot log(\nodes)-2.24$\\
				&$\avgdelay(\nodes)=\textbf{0.43}\cdot log(\nodes)+0.26$\\
				\bottomrule
			\end{tabular}
	}
	\caption{Linear interpolation between $\delay$, $\maxdelay$ (in seconds), and $log(\nodes)$.}
	\label{tab:regression}

	\end{subfigure}
	\hfill
	\caption{Impact of \nodes on \maxdelay and the average delay \avgdelay (in the benign case) using various gossip protocols.}
	\label{fig:delay_n}
\end{figure*}

\paragraph{Supported blockchains}
Our experiments supported the following popular Nakamoto-style blockchains.
To remove any bias incurred by the various popularity of these blockchains (e.g., Bitcoin has far more transactions than Cardano), we consider a fixed block size of 800~KB; this corresponds to the average block size of Bitcoin in 2022.  Unless otherwise specified, we only consider nodes that are actively mining (except for \Cref{tab:zn} where we investigate the impact of \zerominers). Therefore,  our evaluations emulate the realistic scenario where mining pools enjoy dedicated networks between each other (although still restricted to the maximum number of neighbors defined by their respective protocols). We additionally evaluate a ``high-speed`` network connectivity option among mining pools to capture the current setup used by Bitcoin.

\begin{table*}
	\footnotesize
	\centering
		\begin{tabular}{|l|r|c|c|l|}
			\hline
			 & Block interval & Avg. block size  & Max. block size & Gossip protocol\\
			\hline
			Bitcoin~\cite{nakamoto2009bitcoin} & $600$s & 1.7 MB & 2 MB & Low-bandwidth CBR\\
			Ethereum Classic~\cite{ethclassic} & $13$s & 1.2 KB & 30 Mgas & Hybrid push\\
			Monero~\cite{monero} & $120$s & 80 KB & N.A. & Direct push\\
			Cardano~\cite{cardano} & $20$s & 31 KB & 90 KB & Advertisement-based\\
			\hline
		\end{tabular}

	\caption{Summary of the main differences among the studied Nakamoto-style blockchains. The average blocksize was retrieved from~\cite{blocksize}, the remaining parameters are derived from the respective documentation and codebase.}
	\label{tab:params}
\end{table*}

	\paragraph{Bitcoin} Bitcoin emerged as the first real-world deployment of Nakamoto-style consensus. Blocks are generated with a rate of one block every ten minutes on average.
	To disseminate information among nodes, Bitcoin currently relies on the so-called ``compact blocks" mode~\cite{BIP_152} depicted in subplot~(d) of Figure~\ref{fig:gossip}.
	In this work, we focus on the low-bandwidth mode of the compact block mechanism.
	Here, blocks in this mode are sent in a compact format, by including references to the confirmed transactions therein.
	As the reference merely consists of transaction hashes, blocks are compact in size.

	\paragraph{Monero}
	Monero is one of the first altcoin forks of Bitcoin that focused on providing a high degree of anonymity and privacy to the users of the blockchain.
	Here, blocks are propagated by \txtnodes directly to their neighbors. This happens irrespective of whether the recipient has received the block in question in a previous interaction.
	In comparison to \adv propagation, this reduces the number of total messages sent as well as the latency it takes a node to receive a block for the first time.
	However, it comes at the cost of nodes superfluously receiving the same blocks multiple times, which can increase the total overhead in the system.
	Although Monero also provides support for ``compact blocks'' in the high-bandwidth mode, its adoption is not clear, hence we focus on the default direct push mode in this paper.
	Direct push is depicted in subplot~(b) in Figure~\ref{fig:gossip} in the Appendix.

	\paragraph{Cardano} Cardano mostly relies on an \adv push model to propagate transactions and blocks among nodes (this was the initial gossip layer used in Bitcoin).
	Here, upon the receipt of a fresh block, each \txtnode advertises the block to all its outgoing neighbors by sending them an \textit{inv} message.
	This message advertises the knowledge of a new block in the platform.	When a node receives a block advertisement of a new block for the first time, the node explicitly asks to receive the block in question by issuing a \textit{getdata} message.
	Subplot~(a) in Figure~\ref{fig:gossip} depicts the \adv propagation mechanism.

	\paragraph{Ethereum Classic}
	To disseminate information among nodes, Ethereum Classic uses a hybrid approach, combining both aforementioned gossip protocols, to propagate freshly mined blocks to its partaking \txtnodes.
	First, it chooses a subset of its neighbors, e.g., of size the square root of the number of its neighbors (as defined in Ethereum), to whom it directly pushes the block using a \textit{NewBlockMsg} message.
	For the remaining neighbors, the node relies on \adv push to disseminate the block (cf. subplot~(c) of Figure~\ref{fig:gossip}).

Table~\ref{tab:params} summarizes the various differences between the considered Nakamoto-style blockchains.

\subsection{Delays in Large-Scale Blockchains} \label{sec:real_world_delays}

We now proceed to empirically evaluate the block propagation delays in the benign setting (where there are no network-layer attacks) when varying the number \nodes of nodes, the block size, and the gossip protocol.
Here, we emulate the propagation of 100 blocks in networks of size $\nodes \in [10, \num{500000}]$ for each of the deployment instances that we consider i.e., Bitcoin, Ethereum Classic, Monero, and Cardano. 

Our results are summarized in \cref{fig:delay_n}.
When scaling up the network, all considered blockchains exhibit a linear increase in the delays on the semi-log plot in \cref{fig:delay_n}, hence confirming the logarithmic correlation between the maximum delay and the network size \nodes.
Cardano, Monero, and Ethereum Classic tend to have a similar scaling factor for their mean delay \avgdelay, although the slope of the maximal latency exhibited by Cardano is steeper than the others. In contrast, Bitcoin demonstrates significantly lower \avgdelay and \maxdelay, with a correspondingly smaller scaling factor, mostly owing to the reliance on compact blocks.

Based on our empirical result, we built a linear regression between $log(\nodes)$ and $\avgdelay(\nodes)$, $\maxdelay(\nodes)$. To do so, we searched for $\alpha$ and $y_0$ such that $y=\alpha\cdot x + y_0$ minimizes the mean square error using the empirically measured data points. Our results are listed in \cref{tab:regression}. Our results indicate that Bitcoin yields the lowest maximum delays when compared to all other investigated deployments. Notice that the low bandwidth mode of \cbr in Bitcoin behaves similarly\footnote{The slope of \push is 10 times larger than that witnessed by \cbr as shown in Figure~\ref{tab:regression}.} to the \push protocol by Monero (both gossip protocols are very similar albeit with different message sizes); equally, the low bandwidth mode of \cbr behaves similarly to the \adv gossip protocol of Cardano (both protocols are very similar albeit with different message sizes). 

Due to lack of space, we include additional results in Appendix~\ref{ap:gossip_protcols}. Our results suggest that the compact block relay modes emerges as an attractive option to minimize delays as the network grows in size (cf. Figures~\ref{fig:uniform_mean_benign} - \ref{fig:uniform_max_benign}).

\begin{table*}
	\footnotesize
	\centering
	\scalebox{0.95}{
	\begin{tabular}{ |c|c|c|c|c|c|c|c|c|c| }
			\hline
			& Number of&\multicolumn{8}{c|}{Number of mining nodes \nodes}\\
			\cline{3-10}
		 &Zero-power nodes& 200&500&\num{1000}&\num{2000}&\num{5000}&\num{10000}&\num{100000}&\num{500000}\\
		 \hline
		 \multirow{2}{*}{Bitcoin}&\num{5000}&0.84s&1.01s&1.12s&1.20s&1.17s&1.14s&1.13s&1.24s\\
		 &\num{100000}&1.07s&1.34s&1.47s&1.56s&1.73s&1.83s&1.86s&1.86s\\
		 \hline
		 \multirow{2}{*}{Cardano}&\num{5000}&13.86s&15.57s&17.10s&18.55s&20.54s&22.78s&26.60s&31.43s\\
		 &\num{100000}&21.82s&16.64s&16.55s&17.73s&19.11s&20.42s&22.93s&26.33s\\
		 \hline
		 {Ethereum }&\num{5000}&7.79s&9.39s&10.73s&11.66s&13.01s&14.32s&17.03s&21.32s\\
		 Classic&\num{100000}&13.83s&10.14s&11.08s&12.70s&14.40s&16.04s&18.07s&20.88s\\
		 \hline
		 \multirow{2}{*}{Monero}&\num{5000}&8.86s&10.88s&11.69s&11.25s&10.28s&10.61s&10.45s&11.59s\\
		 &\num{100000}&10.06s&10.30s&12.31s&14.31s&16.30s&18.07s&18.61s&18.54s\\
		 \hline
	\end{tabular}}
	\caption{Impact of \txtzps on \maxdelay in seconds.}
	\label{tab:zn}
\end{table*}

Finally, we analyzed the impact of zero-power nodes, i.e., nodes that maintain a complete copy of the deployment's ledger but are not actively participating in generating blocks in~\Cref{tab:zn}.
Here, we do not explicitly enforce short connections between \txtminers such that \txtminers constitute the network's core, but instead \txtminers are distributed throughout the entire network randomly.

The impact of \txtzps on the \maxdelay for small values of \nodes is substantial---adding up to approximately 8 seconds to the overall maximum latency witnessed in the system, corresponding to an increase of 56\% for Cardano, 77\% for Ethereum Classic, 13.5\% for Monero, and 25\% for Bitcoin.
On the other hand, for larger values of \nodes, the increase of \txtzps ensures an overall better network connectivity reducing the \maxdelay by 4.86s in Cardano and 0.44s in Ethereum Classic.
For Monero and Bitcoin, increasing the number of \txtzps up to \num{100000} nodes however results in an increase of up 78\% and 50\%, respectively, in \maxdelay.

\begin{figure}
	\centering
	\includegraphics[width=0.67\linewidth]{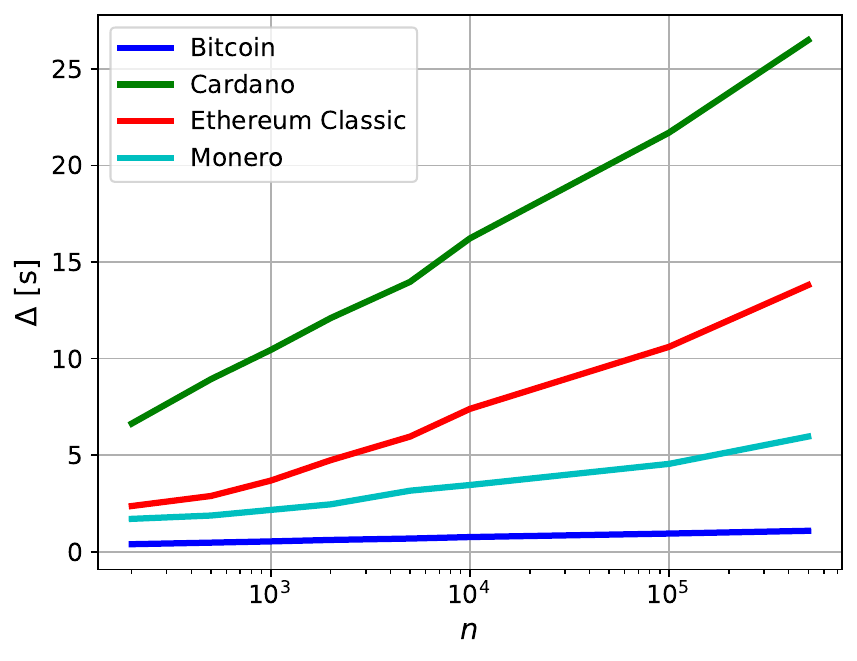}
	\caption{Impact of dedicated high-speed networks among miners on \maxdelay.}
	\label{fig:high_bandwidth_max}
\end{figure}

\begin{table}
\footnotesize
\centering
\setlength\tabcolsep{3pt}
\scalebox{0.95}{\begin{tabular}{ |c|c|c|c|c|c| }
	\hline
	\multirow{3}{*}{} & \multirow{3}{*}{$\constDelay$} & \multicolumn{4}{c|}{Adversarial power $\poweradvmax$} \\
	\cline{3-6}
	 &  & \multicolumn{4}{c|}{Number of \txtminers \numminers} \\
	 \hline
	 && $10^1$ & $10^3$ & $10^6$ & $10^9$ \\
	\hline
	\hline
	\multirow{4}{*}{Bitcoin} & 0~s (no network attacks) & 0.5000 & 0.4999 & 0.4997 & 0.4996 \\
	 & 600~s ($\sim$1~delayed block) & 0.3819 & 0.3819 & 0.3817 & 0.3816 \\
	 & \num{3000}~s ($\sim$5~delayed blocks) & 0.1615 & 0.1615 & 0.1614 & 0.1614 \\
	 & 1Gb/s Dedicated network& 0.5000 & 0.4999 & 0.4998 & 0.4997\\

	\hline
	\multirow{4}{*}{Monero} & 0~s (no network attacks) & 0.4995 & 0.4938 & 0.4854 & 0.4769 \\
	 & 120~s ($\sim$1~delayed block) & 0.3815 & 0.3768 & 0.3698 & 0.3630 \\
	 & 600~s ($\sim$5~delayed blocks) & 0.1614 & 0.1603 & 0.1586 & 0.1570 \\
	 	 & 1Gb/s Dedicated network& 0.4999 & 0.4973 & 0.4934 & 0.4897\\
	\hline
	\multirow{4}{*}{Cardano} & 0~s (no network attacks) & 0.4935 & 0.3935 & 0.2820 & 0.2135 \\
	 & 20~s ($\sim$1~delayed block)& 0.3765 & 0.3010 & 0.2251 & 0.1776 \\
	 & 100~s ($\sim$5~delayed blocks) & 0.1602 & 0.1417 & 0.1206 & 0.1049 \\
	 	 & 1Gb/s Dedicated network& 0.5000 & 0.4344 & 0.3416 & 0.2734\\
	\hline
	\multirow{4}{*}{\shortstack[c]{Ethereum \\ Classic}} & 0~s (no network attacks) & 0.5000 & 0.4155 & 0.2896 & 0.2142 \\
	 & 13~s ($\sim$1~delayed block) & 0.4042 & 0.3166 & 0.2302 & 0.1782 \\
	 & 65~s ($\sim$5~delayed blocks) & 0.1668 & 0.1456 & 0.1222 & 0.1051 \\
	 & 1Gb/s Dedicated network& 0.5000 & 0.4560 & 0.3703 & 0.3026\\
	\hline
\end{tabular}}
\caption{Maximum adversarial power that can be tolerated in real-world deployments w.r.t. the number of nodes.}
\label{tab:HMA}
\end{table}

\begin{figure*}[htbp]
    \centering
    \begin{minipage}[b]{0.63\linewidth}
        \centering
        \begin{subfigure}[t]{0.49\textwidth}
		\includegraphics[width=\linewidth]{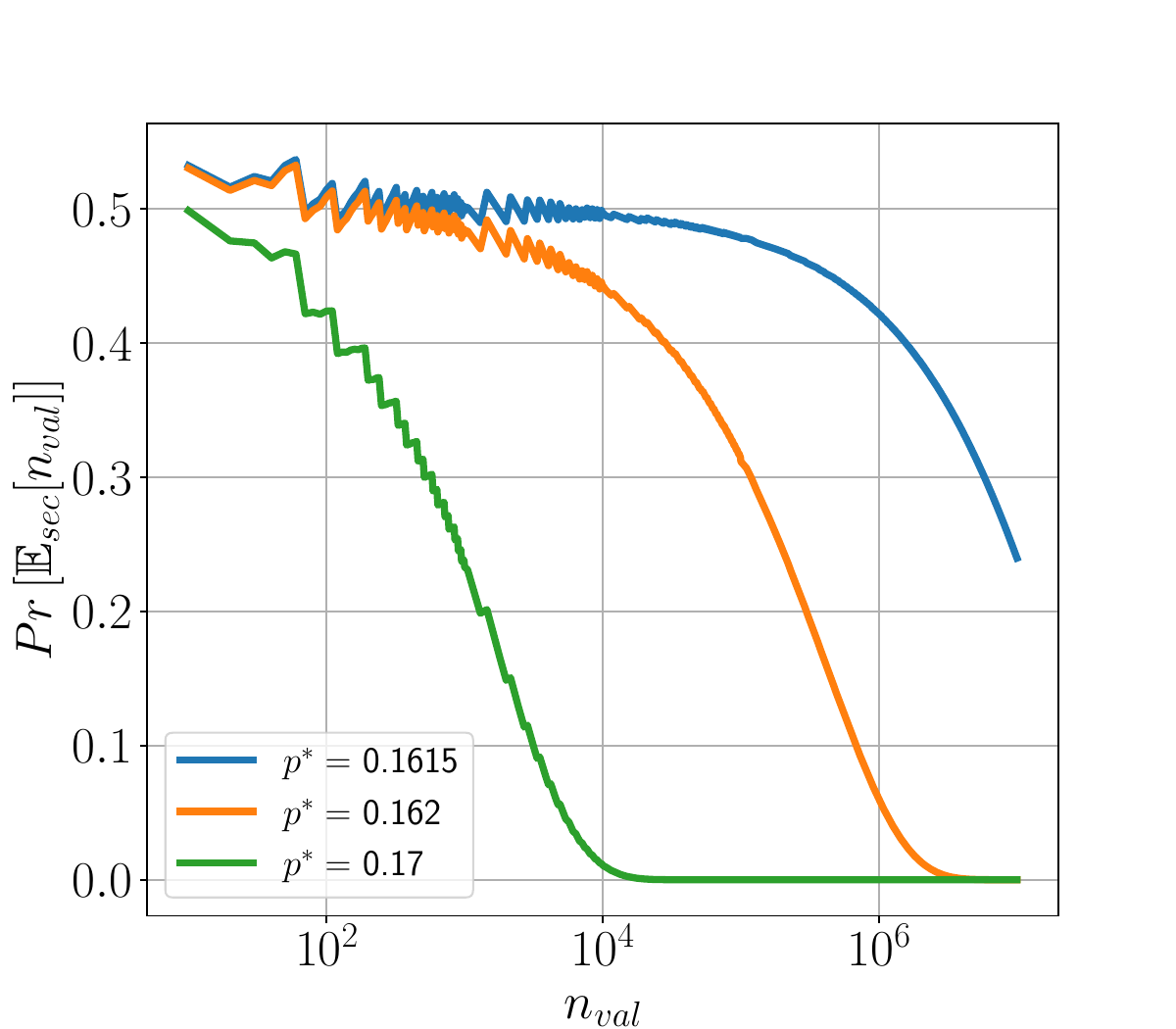}
		\caption{The case of {Bitcoin}.}
		\label{fig:BC_theoretical}
	\end{subfigure}
	\begin{subfigure}[t]{0.49\textwidth}
		\includegraphics[width=\linewidth]{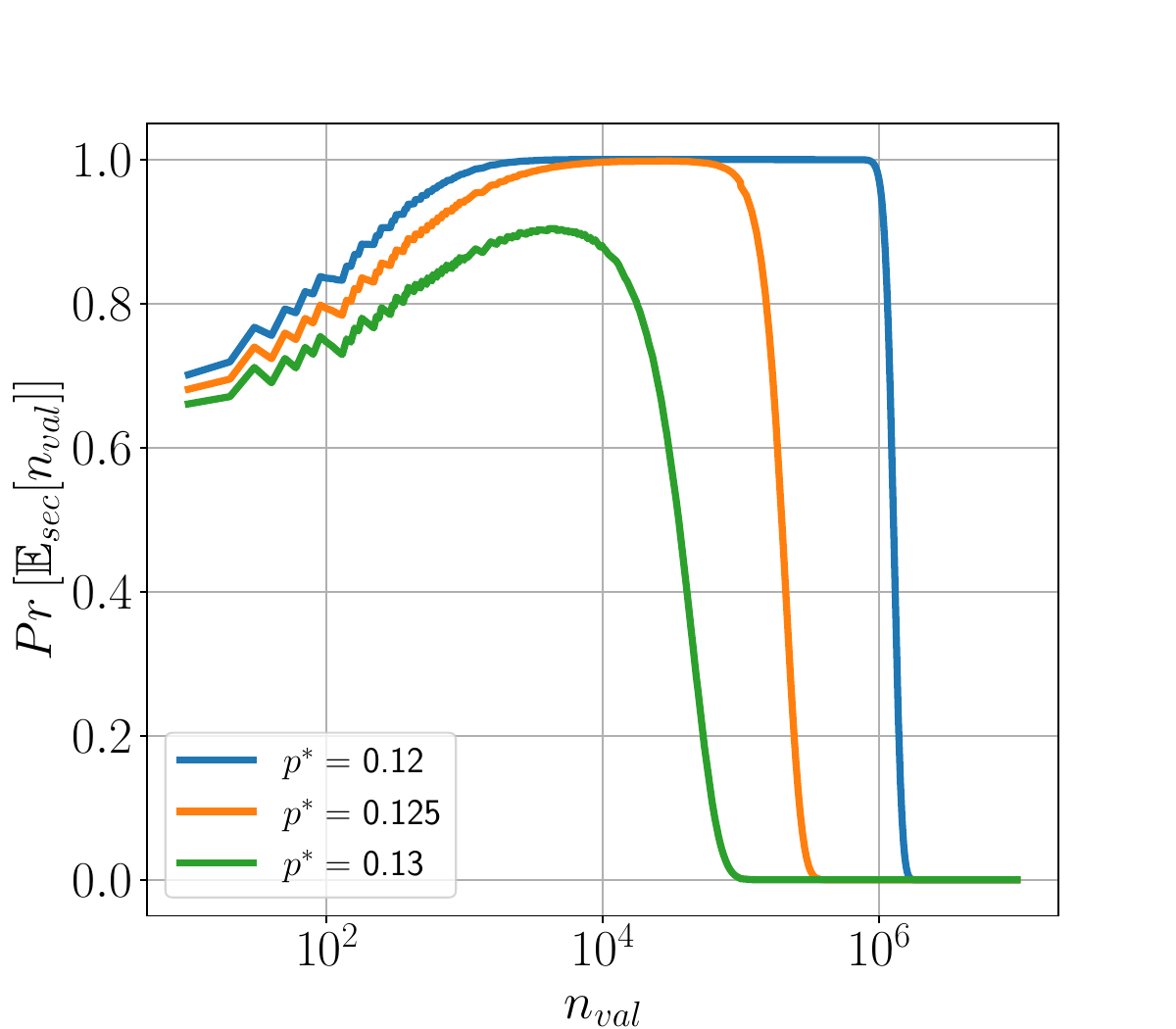}
		\caption{The case of {Cardano}.}
		\label{fig:CA_theoretical}
	\end{subfigure}
	\begin{subfigure}[t]{0.49\textwidth}
		\includegraphics[width=\linewidth]{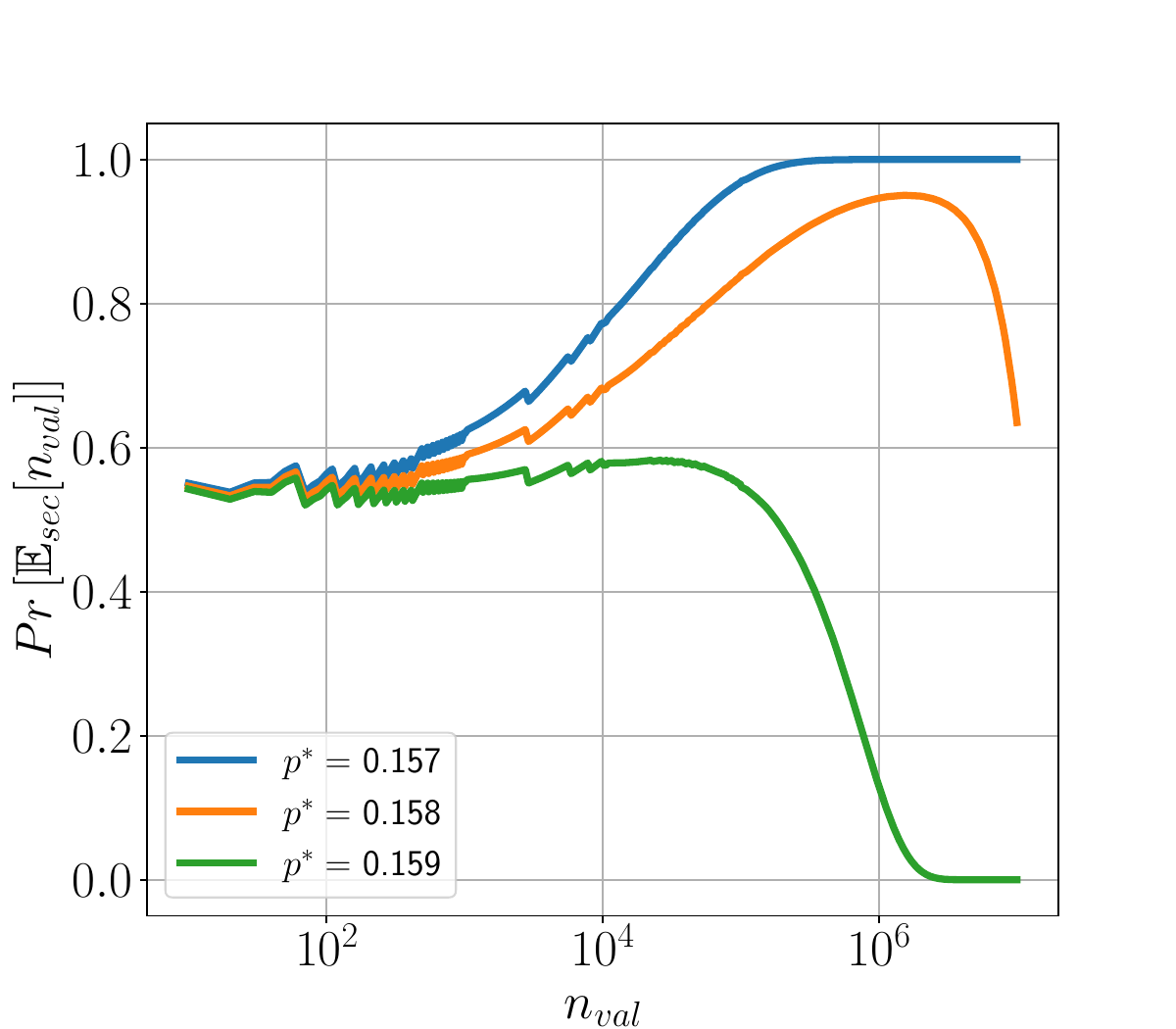}
		\caption{The case of {Monero}.}
		\label{fig:MO_theoretical}
	\end{subfigure}
	\begin{subfigure}[t]{0.49\textwidth}
		\includegraphics[width=\linewidth]{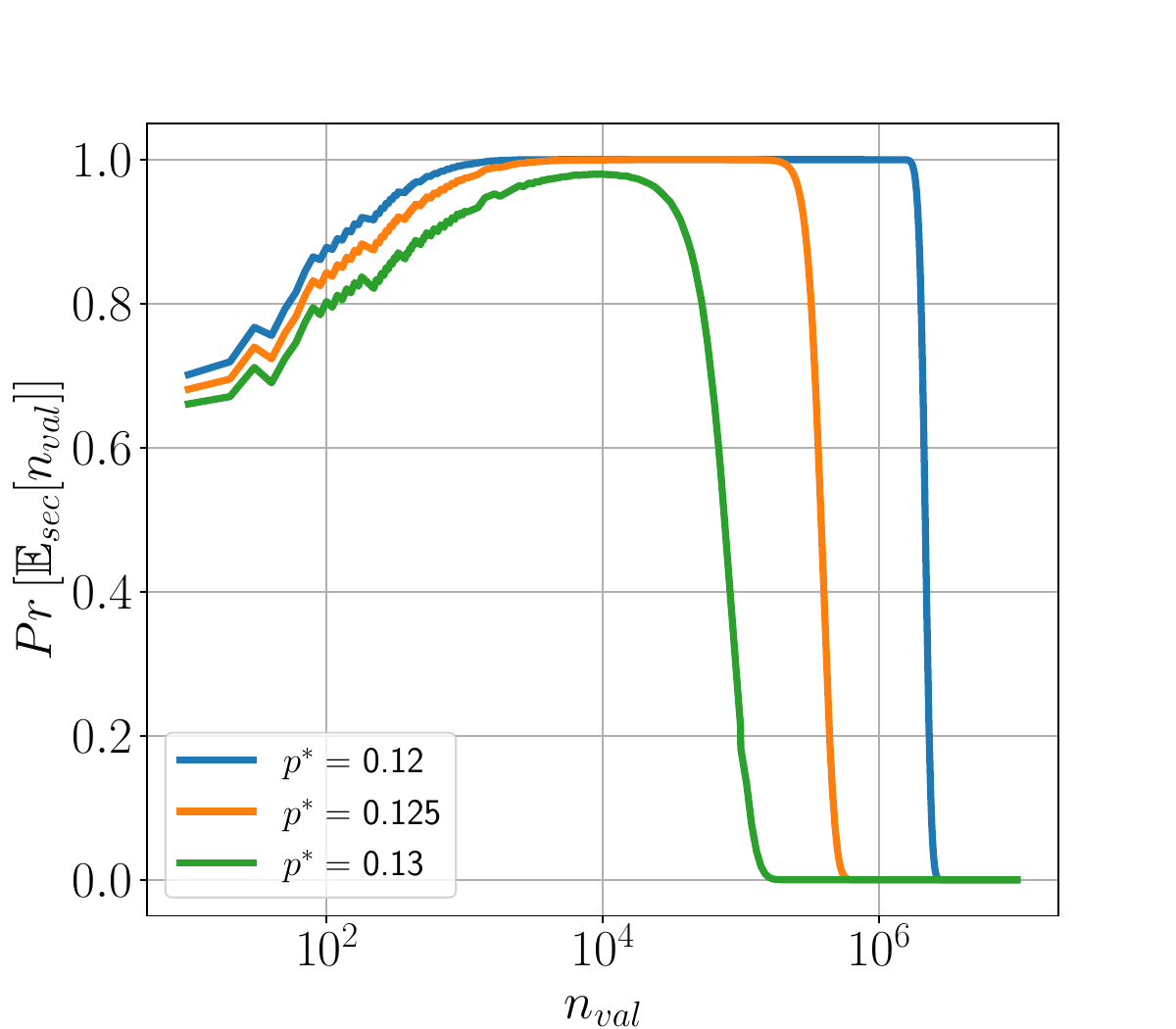}
		\caption{The case of {Ethereum Classic}.}
		\label{fig:ET_theoretical}
	\end{subfigure}
	\caption{Evaluating the probability $\prob{\secevent[\numminers,\numzp]}$ that the blockchain provides persistence and liveness depending on the network's size. Here, we assume an adversary that can selectively delay at most 5 consecutive blocks among a fraction of nodes.}
	\label{fig:theoretical_bounds}
    \end{minipage}
    \hspace{0.02\linewidth} 
    \begin{minipage}[b]{0.309\linewidth}
        \centering
        \begin{subfigure}[t]{\textwidth}
		\includegraphics[width=\linewidth]{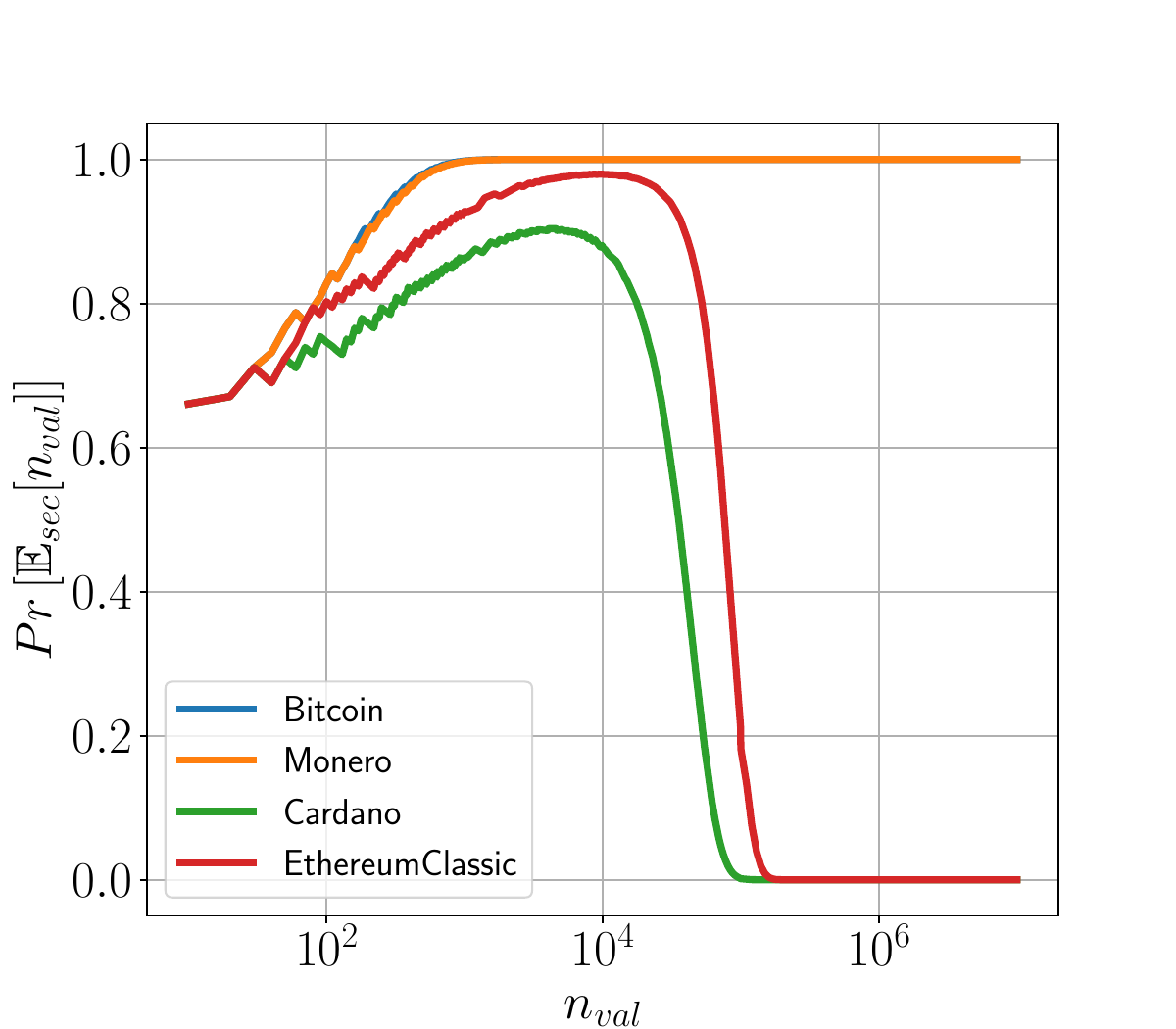}
		\caption{The case of $\EV = 0.13$.}
		\label{fig:theoretical_bound_13}
	\end{subfigure}
	\begin{subfigure}[t]{\textwidth}
		\includegraphics[width=\linewidth]{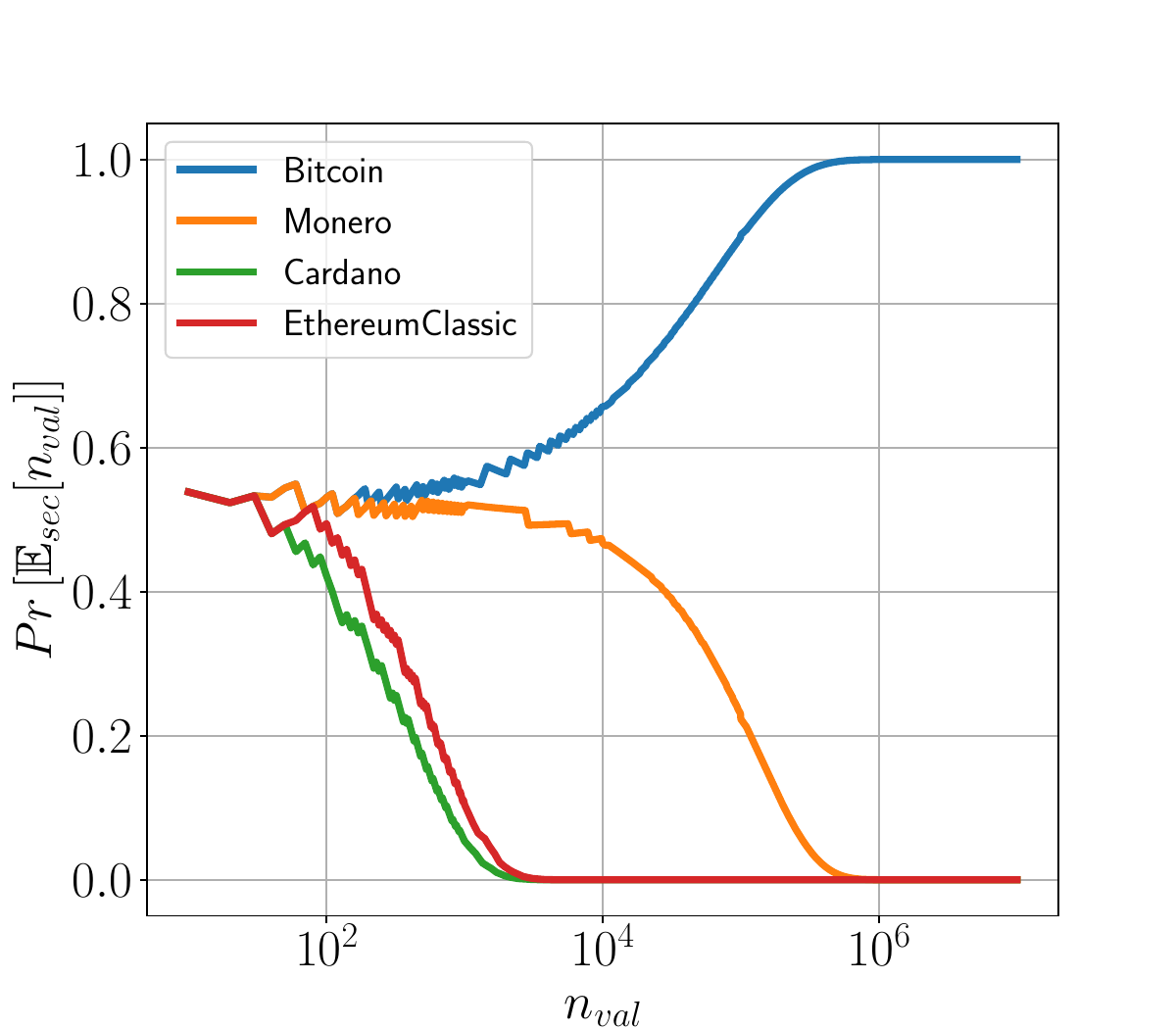}
		\caption{The case of $\EV = 0.16$.}
		\label{fig:theoretical_bound_16}
	\end{subfigure}
        \caption{Evaluating  $\prob{\secevent[\numminers,\numzp]}$ for Bitcoin, Monero, Cardano, and Ethereum Classic when  $\EV = 0.13$ and $0.16$.}
        \label{fig:theoretical_bounds_comp}
    \end{minipage}
\end{figure*}

\begin{figure*}[tb]
	\centering
		\begin{minipage}[c]{0.06\textwidth}
			\centering
				\rotatebox{90}{
			\centering
				\fbox{
		\begin{subfigure}[t]{4\textwidth}
			\includegraphics[width=\linewidth,trim={0.1cm 2.0cm 0.1cm 0.1cm},clip]{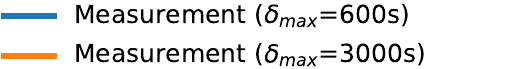}
		\end{subfigure}
		\begin{subfigure}[t]{4\textwidth}
		\includegraphics[width=\linewidth,trim={0.1cm 0.6cm 0.1cm 1.5cm},clip]{legendadv.pdf}
	\end{subfigure}
	}
	}
	\end{minipage}
	\begin{minipage}[c]{0.7\textwidth}
	\centering
	\begin{subfigure}[t]{0.49\textwidth}
		\includegraphics[width=\linewidth]{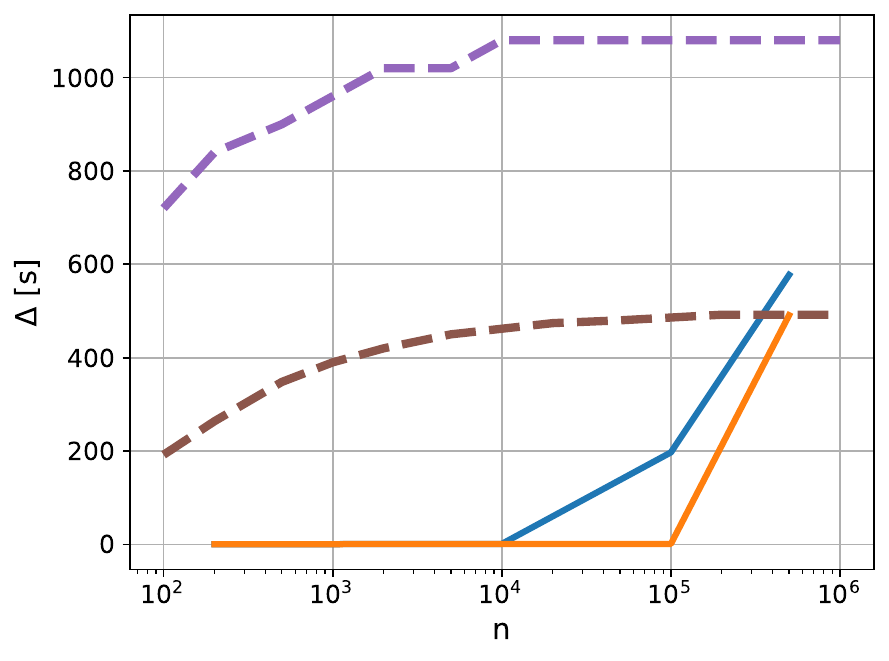}
		\caption{The case of {Bitcoin}. }
		\label{fig:BC_comparison}
	\end{subfigure}
	\begin{subfigure}[t]{0.49\textwidth}
		\includegraphics[width=\linewidth]{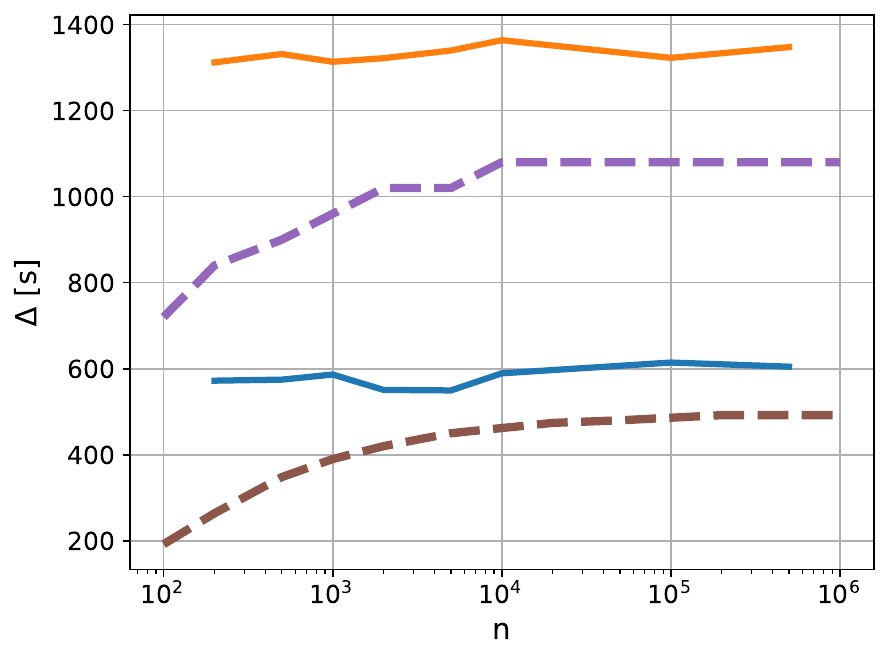}
		\caption{The case of {Cardano}. }
		\label{fig:CA_comparison}
	\end{subfigure}
	\begin{subfigure}[t]{0.49\textwidth}
		\includegraphics[width=\linewidth]{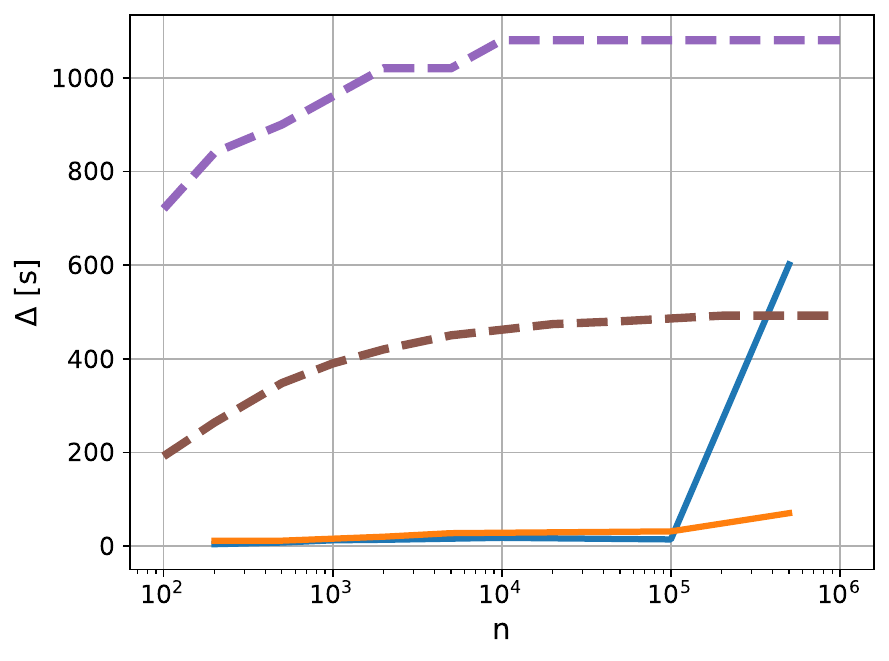}
		\caption{The case of {Monero}. }
		\label{fig:MO_comparison}
	\end{subfigure}
	\begin{subfigure}[t]{0.49\textwidth}
		\includegraphics[width=\linewidth]{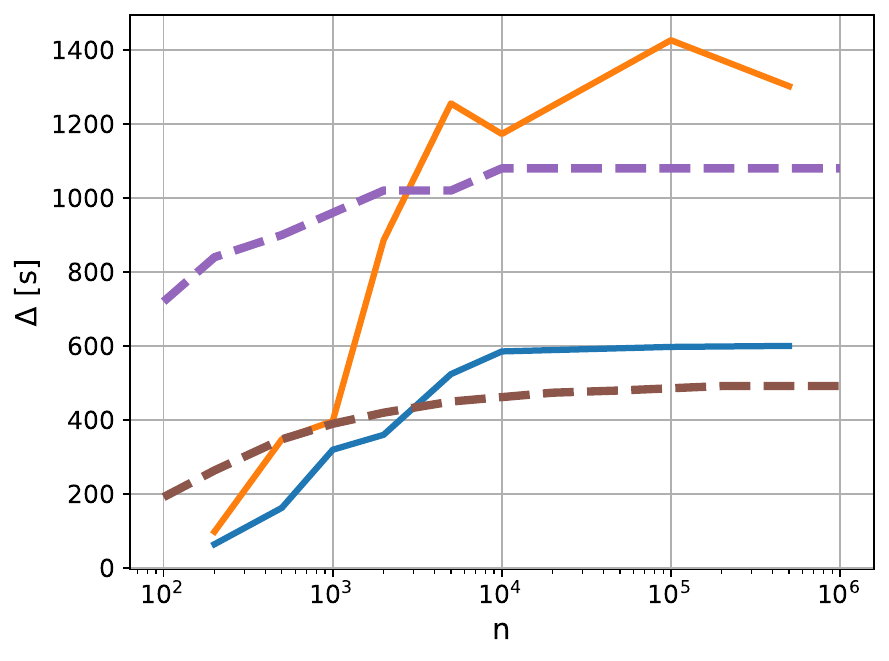}
		\caption{The case of {Ethereum Classic}.}
		\label{fig:EC_comparison}
	\end{subfigure}
	\end{minipage}
	\caption{ Comparing analytical bounds (dashed lines) and empirical measurements (solid lines) obtained for \maxdelay. Here, \EV denotes the probability for a \txtminer being corrupted and with an adversary being present, who controls 15\% of the network's nodes and induces delays on 50\% of his connections of each $nt_{delay}$ seconds (cf.\ \cref{sec:network_attack}).}

	\label{fig:empirical_bounds}

\end{figure*}

\section{Tradeoffs in Real-World Deployments} \label{sec:real_world_analysis}

We now turn to analyzing the impact of the network scale on the security of Nakamoto-style deployments.

\subsection{Fraction of Adversarial Power} \label{sec:real_world_advpower}

We start by investigating the impact of the network size on the fraction of adversarial power that can be tolerated. As mentioned in Section~\ref{sec:gaps}, prior work commonly assumes that the tolerated adversarial power is bounded by a \emph{fixed} value, e.g., $\poweradvmax < \frac{1}{3}$ or $\poweradvmax < \frac{1}{2}$.  However, as shown in Equation~(\ref{eq:adv_bound}), the maximum adversarial power depends on the number of nodes \nodes, the block frequency \rate, and the delay \maxdelay.

Notice that, in real-world settings, the adversary is able to carry out network-layer attacks~\cite{DBLP:conf/ccs/GervaisRKC15, DBLP:conf/uss/HeilmanKZG15} to manipulate the flow of information or even eclipse some nodes. To capture this behavior,
we introduce the term $\constDelay$ for the upper bound on the (adversarially induced) delay that is caused by such attacks. In other words, we approximate $\maxdelay(\nodes)$ with
\begin{equation}
	\maxdelay(\nodes, \constDelay) = \underbrace{\left( \alpha \cdot \log(\nodes)+ y_0 \right)}_{\text{cf.\ \cref{tab:regression} }} + \constDelay \label{eq:adv_delays}
\end{equation}
	 
For instance, to capture the network attacks due to Gervais et al.~\cite{DBLP:conf/ccs/GervaisRKC15} and Heilman et al.~\cite{DBLP:conf/uss/HeilmanKZG15}, $\constDelay$ can be set to 50~mins.\footnote{This corresponds to delaying approximately 5 consecutive blocks with a success rate of~0.4~\cite{DBLP:conf/ccs/GervaisRKC15}.}

In Table~\ref{tab:HMA}, {we report the values of~$\poweradvmax$ for real-world deployments of Nakamoto-style blockchains, namely for Bitcoin, Monero, Cardano, and Ethereum Classic, 
where we estimated~$\poweradvmax$ by plugging the empirical block propagation times we measured for each deployment}
(cf.\ Figure~\ref{tab:regression}) into Equation~(\ref{eq:max_adv_power}). 

Our results confirm that the maximum adversarial power that can be tolerated is strongly influenced by the scale of the network as well as by network-layer attacks. For instance, while Ethereum Classic with $10$ validators can tolerate \poweradvmax of 50.0\%, this bound decreases to $29.0$\% for networks with \num{1000000} validators without any active (network-layer) attack; here, if an attacker is able to manipulate information at the network layer and selectively delay single blocks, the bound decreases to $12.2$\%.
A similar behavior can be seen for Cardano.
{In contrast, the tolerated adversarial power in Bitcoin and Monero appears to be only mildly affected by the network size (see Figure~\ref{fig:delay_n}). This phenomenon shows that the choice of the gossip protocol is crucial for the security of the system.}

	\paragraph{Impact of Dedicated High-Speed Networks}
We now evaluate the impact of having a dedicated high-speed network between miners, similar to the current setup used by mining pools, on the maximum fraction of adversarial power that can be tolerated. To this end, we equip each pool with a 1Gb/s download/upload network with 10 ms latency.
We empirically evaluate the \maxdelay for all studied blockchains in this setting and plot the corresponding results in \Cref{fig:high_bandwidth_max}. Same as before, we also compute the regression of \maxdelay with respect to $log(\nodes)$ analogously to \Cref{tab:regression}.
As expected, a dedicated high-speed network significantly reduces the scaling factor of the \maxdelay in terms of $\log(\nodes)$, up to 20\% for Bitcoin, 30\% for Cardano, 45\% for Ethereum classic, and 54\% for Monero.
Bitcoin exhibits the smallest latency reduction due to the inherently low latency, resulting in the constant block verification time becoming the predominant contributor to the propagation delay instead of the network latency.

While the improvement in \maxdelay has a positive impact on security, it does not completely compensate for the impact of the scale of the network. Namely, for the same number of \nodes, Bitcoin (resp. Cardano, Ethereum Classic, Monero) can tolerate only up to 0.01\% (resp. 1.3\%, 6.01\%, 8.84\%) additional adversarial power $\poweradvmax$ as shown in \Cref{tab:HMA}. For instance, in Cardano and Ethereum classic, even when assuming a high-speed dedicated network, when the number of validator \nodes increases to \num{1000000}, the maximum adversarial power that can be tolerated is only 34.16\%, and  37.03\% respectively.
Note that when the number of miners (aka mining pools) is small (e.g., 10), this fraction is much higher---even when the mining pools do not use a high-speed dedicated network and merely use the underlying gossip layer.

\subsection{Scale vs Security}
\label{sec:network_attack}

In what follows, we evaluate the impact of the network size on $\prob{\secevent[\numminers]}$, the probability that 
{the actually held adversarial power \poweradv is below the maximum tolerated threshold \poweradvmax (i.e., the trust assumption that security guarantees are met).} 

For our analysis, we consider the case where $\constDelay$ is bounded by the time to selectively delay the propagation of up to 5 consecutive blocks~\cite{DBLP:conf/ccs/GervaisRKC15, DBLP:conf/uss/HeilmanKZG15}. We then rely on our empirical measurements to compute the overall delay $\maxdelay(\nodes, \constDelay)$ (cf.\ Equation~(\ref{eq:adv_delays})) and plug the resulting value in Equation~(\ref{eq:success}) to derive the corresponding value for $\prob{\secevent[\numminers]}$. In particular, we first computed $g(n) = \lfloor \poweradvmax\cdot \numminers \rfloor$ as defined by Equation~(\ref{eq:success}) and subsequently determined the security probability as $\prob{\secevent[\numminers]} = F(g(n); n, \EV)$ using the cumulative distribution function $F(k; n, p)$ of the binomial distribution.\footnote{
For instance, to derive the plot for $\EV = 0.17$ from \cref{fig:BC_theoretical}, we computed 
\begin{equation}
	g(n) = \left\lfloor \nodes \cdot \left( \frac{\varMainTheorem + 1 + e}{2\varMainTheorem} - \sqrt{\left(\frac{\varMainTheorem + 1 + e}{2\varMainTheorem}\right)^2 - \frac{1}{\varMainTheorem}} \right) \right\rfloor
\end{equation} 
using $\varMainTheorem = \frac{1}{600} \cdot \maxdelay(\nodes, \constDelay)$ with $\maxdelay(\nodes, \constDelay) = (0.1 \cdot \log(n) - 0.04) + 300$ (cf.~Equation~(\ref{eq:adv_bound}), \cref{eq:adv_delays}, and \cref{tab:regression}). Then, we evaluated $F(g(n); n, 0.17)$ over $\nodes \in [10, 10^7]$.
}

Our results are depicted in Figure~\ref{fig:theoretical_bounds} and Figure~\ref{fig:theoretical_bounds_comp}. We generally observe that an increase in the number of validators \numminers increases the probability of the system being secure until a turnaround value is reached. This turnaround value depends on several parameters:
\begin{itemize}
\item the probability \EV for a \txtminer being corrupted,
\item the maximum delay $\maxdelay$,
\item the block rate \rate.
\end{itemize}

{Our results show that a smaller corruption probability $\EV$ results in higher turnaround values for \numminers. We justify this behavior by observing that it becomes harder for the attacker to gain an amount of power exceeding \poweradvmax as corrupting individual validators gets more difficult.}
{Looking at the the plot for $\EV = 0.157$ in case of Monero, the turnaround point is even outside of the plotted range in Figure~\ref{fig:MO_theoretical}.
On the other hand, if we evaluate a high $\EV \geq 0.16$ in the case of Bitcoin, the turning point is reached already for $\nodes \leq 10$ nodes, and we do not observe an increase in security in Figure~\ref{fig:BC_theoretical} since the turning point is for $n \leq 10$ (which is outside the plotted range).
We finally note that for comparable values of $\EV$, Bitcoin and Monero consistently presents a higher turnaround point compared to Cardano and Ethereum Classic (cf.~\cref{fig:theoretical_bound_13}). Further, Bitcoin provides a higher turnaround point compared to Monero (cf.~\cref{fig:theoretical_bound_16}). These results confirm that Bitcoin and Monero scale better with respect to the network size compared to Cardano and Ethereum Classic (because of the delays incurred by their respective gossip protocols).

{In the next set of experiments we compare the analytical upper bound on the maximum delay \maxdelay with empirically measured delays in common deployments. The results are shown in Figure~\ref{fig:empirical_bounds}.
To this end, we determine the maximum delay $\maxdelay_{max}$ fulfilling the condition in Equation~(\ref{eq:success}), as follows:
\begin{equation*}
	\maxdelay_{max} := \max \{\maxdelay \in \mathbb{R}\ \vert \ \prob{\secevent[\numminers,\numzp]} \geq 0.9 \}
\end{equation*}
}

Here, in our experiments we assume that the adversary can compromise each node with probability 15\% to execute network attacks.
Further, each such malicious node can delay its messages to a fraction of 50\% of its connected neighbors by at most 10 and 50 minutes respectively. We include the evaluation results pertaining to various network attack strategies in Appendix~\ref{ap:gossip_protcols}.

{We observe that Bitcoin and Monero are only mildly affected by network delays thanks to their usage of \cbr and \push, which makes them robust against our strongest adversaries of $\EV = 0.4$ for smaller networks (\nodes $\leq$ \num{300000}), and $\EV = 0.3$ for larger networks (\nodes $>$ \num{300000}). }
On the other hand, we observe that Cardano's \adv protocol is very susceptible to network delays: after selecting an inv message to retrieve the block from, nodes wait for a timeout of 10 minutes before requesting the block from another node. Therefore, even for a network delay of 600s with \nodes = 200, Cardano is not resilient to the $\EV=0.4$ adversary and is further susceptible to the $\EV=0.3$ adversary when the network delay increases to 50 minutes.
Lastly, we notice that while Ethereum Classic's \hybrid provides good performance for small networks, security degrades as \nodes grows, making it similarly susceptible to the $\EV=0.4$  and $\EV=0.3$ adversaries when \nodes$\geq$\num{5000} for a delay of 10 minutes and 50 minutes respectively.

\section{Related Work}

The honest-majority assumption, initially formulated as an intuitive requirement for the security of Bitcoin, has been gradually refined in a series of theoretical works seeking to precisely capture the impact of network delays on consistency~\cite{DBLP:conf/eurocrypt/GarayKL15, DBLP:journals/iacr/KifferRS22, DBLP:conf/icdcs/ZhaoTLWLX20, DBLP:journals/corr/abs-2203-06357, DBLP:conf/ccs/DemboKTTVWZ20}, leading to more expressive (and tighter) security conditions for Nakamoto-style blockchains (see Table~\ref{tab:conditions} for a detailed comparison).

Independently from these models, a number of contributions reported practical attacks by large-scale adversaries on Nakamoto-style deployments~\cite{DBLP:conf/uss/HeilmanKZG15} and even by resource-constrained adversaries~\cite{DBLP:conf/ccs/GervaisRKC15}. For instance, in~\cite{DBLP:conf/sp/ApostolakiZV17}, it was shown that routing attacks can be leveraged to isolate a set of nodes and partition the network or to also delay the block propagation in Bitcoin. These attacks allow the adversary to arbitrarily delay the propagation of blocks and transactions in the network by dozens of minutes. In this work, we show that such attacks could easily violate the consistency bounds in a number of existing blockchain deployments.

	Gervais \emph{et al.}~\cite{DBLP:conf/ccs/GervaisKWGRC16} investigated the security and performance of PoW blockchains with respect to various consensus and network parameters. 
	More precisely, they devised optimal strategies for double-spending and selfish mining attacks as a means to quantify robustness against these attacks with respect to the attained overall performance.
	
	Zhang \emph{et al.}~\cite{DBLP:conf/ndss/0003ZWWXP22} investigated tradeoffs between security and performance exhibited by the compact block relay gossip protocol~\cite{BIP_152}. One interesting finding is that fresh transactions, i.e., transactions that are already incorporated in a block but have not been already received by other nodes, increase the effective delay it takes to propagate a newly-mined block in the network (and therefore reduce the overall performance).

Finally, a series of contributions~\cite{DBLP:journals/iacr/GaziKR23,DBLP:journals/corr/abs-2110-08673,DBLP:conf/ccs/DavidM0NT22} have investigated the security of the committee selection in (non-Nakamoto-style) proof-of-stake blockchains. 
{In this vein, Gazi \emph{et al.}~\cite{DBLP:journals/iacr/GaziKR23} proposed a novel committee-selection scheme that leverages specific stake distributions to improve the tradeoff between security and committee size---here, \emph{security} means that the fraction of corrupted validators in each committee should not be ``too far'' from the overall adversarial power. 
In particular, all security analyses of committee-based protocols \emph{assume} that the overall adversarial power is below $\poweradv$. In contrast, our work introduces an alternative probabilistic corruption model capturing the intuition that controlling~$\poweradv\cdot \nodes$ validators becomes harder (i.e., less likely) as $\nodes$ grows.

}

\section{Concluding Remarks}\label{sec:takeaways}

In this work, we investigated how the network scale (i.e., the number of nodes~\nodes) affects the security of Nakamoto-style blockchains in real-world deployments. We specifically analyzed, both theoretically and empirically, how increasing decentralization (i.e., larger \nodes) can result in higher delays (\maxdelay) and higher resilience to Byzantine validators (\poweradvmax). By establishing precise relations between \nodes, \maxdelay and \poweradvmax, we believe that our results lay clear foundations to better explore security vs performance vs decentralization tradeoffs in Nakamoto-style blockchain deployments. 

From an analytical perspective, we observe that while \maxdelay only grows logarithmically in \nodes, the probability of violating the honest-majority assumption decreases exponentially in \nodes. This suggests that, at least asymptotically, decentralization almost always improves security.

Our empirical results show however that in practice, network delays can increase dramatically with \nodes depending on the gossip protocol in use. For instance, Nakamoto-style deployments using low-bandwidth CBR and direct push (i.e., Bitcoin and Monero) appear to be considerably more robust than those relying on hybrid-push and advertisement-based gossip protocols (Ethereum Classic and Cardano) (cf.~\cref{tab:regression} and \cref{fig:empirical_bounds}). In particular, our results indicate that the choice of the gossip protocol is crucial: 
	{For instance, \cref{tab:HMA} shows that both Cardano and Ethereum Classic are vulnerable to an adversary that controls more than 29\% of the total network's power in networks comprising a million nodes. On the contrary, due to their gossip protocols, Bitcoin as well as Monero provide 
	 high security, tolerating a 47\%, resp.~49\% adversary, even in larger networks of up $10^9$ nodes.}

Our preliminary results also indicate that the compact block propagation mode can promise some robustness in spite of a network adversary that can compromise a fraction of the underlying network. We therefore hope that our findings motivate further research in the area.

\section*{Acknowledgment}
This work has been partially funded by the Deutsche Forschungsgemeinschaft (DFG, German Research Foundation)---EXC 2092 (CASA) 39078197---and the European Union (HORIZON-JU-SNS-2022 NANCY project No~101096456 and HORIZON-CL4-2022 Programming Platform for Intelligent Collaborative Deployments over Heterogeneous Edge-IoT
Environments project No 101093069). 
The views and opinions expressed are those of the authors only and do not necessarily reflect those of the European Union or the European Research Council Executive Agency.
Neither the European Union nor the granting authority can be held responsible for them.

\bibliographystyle{IEEEtran}
\bibliography{main_bib}

\appendices
%\section{Additional Proofs}
\section{Proof of Theorem~\ref{th:maxdelay}}\label{ap:thmaxdelay}

	In Nakamoto-style blockchains, nodes randomly join and leave the network and might also change the set of peers over time. That is, the structure of the network results from a random process and the blockchain network can be considered as a random graph~\cite{DBLP:journals/tem/ShahsavariZT22, DBLP:conf/icdcn/CrucianiP23, DBLP:conf/soda/BecchettiCNPT20}. Moreover, we assume that each newly added node establishes up to \outdegree outgoing connections to a set of nodes, where \outdegree is independent of the number of nodes. This is conforming with existing popular Nakamoto-style deployments, such as Bitcoin and Ethereum 1.0.

	The probability $p$ that a fixed node is connected to an arbitrary other node is $\approx\frac{\outdegree}{\nodes}$. A consequence is that $\nodes\cdot p\approx \outdegree$, being a constant $> 1$. 
	
	For the proof, we make use of the following result about the diameter of random graph:
\begin{theorem}[{\cite[Theorem 6]{CHUNG2001257}}]\label{th:diameter}
	Consider a random graph, i.e., a randomly generated graph, of \numnodes \txtnodes in which a pair of
	vertices appears as an edge with probability $p$. Let \diameter denote the diameter of the graph, i.e., the length of the longest shortest path in the graph.
	Suppose that $\nodes \cdot p \geq  c > 1$ for some constant $c$. Then, we have that:
	\begin{eqnarray*}
		\diameter&\geq &(1+o(1))\frac{\log(\nodes)}{\log(n\cdot p)}\ \mathrm{and}\\ 
		\diameter&\leq &\frac{\log(\nodes)}{\log(n\cdot p)}+2\frac{(10c/(\sqrt{c}-1)^2+1)}{c-\log(2c)}\frac{\log(\nodes)}{\log(n\cdot p)}+1.
	\end{eqnarray*}

\end{theorem}
The following corollary is an immediate consequence of \cref{th:diameter}.
\begin{corollary}\label{cor:diameter}
	Consider the setting explained in \cref{th:diameter}. If $\nodes\cdot p$ is constant, it follows that:
	\begin{equation}
		\diameter\in\Theta(\log(\nodes)).
	\end{equation}
\end{corollary}

	Then, we can apply \cref{th:diameter} and \cref{cor:diameter} to derive
	\begin{equation}
		\diameter\in\Theta(\log(\numnodes)).
	\end{equation}
	
	Moreover, the maximum delay \maxdelay is given by the time required for sending a 1-hop-message times the maximum number of hops between any two nodes. As the latter is the value \diameter, it follows that
	\begin{equation}
		\maxdelay=const\cdot \diameter \in const\cdot \Theta(\log(\numnodes))=\Theta(\log(\numnodes)).
	\end{equation}

\begin{figure}
	\centering
	\includegraphics[width=0.67\linewidth]{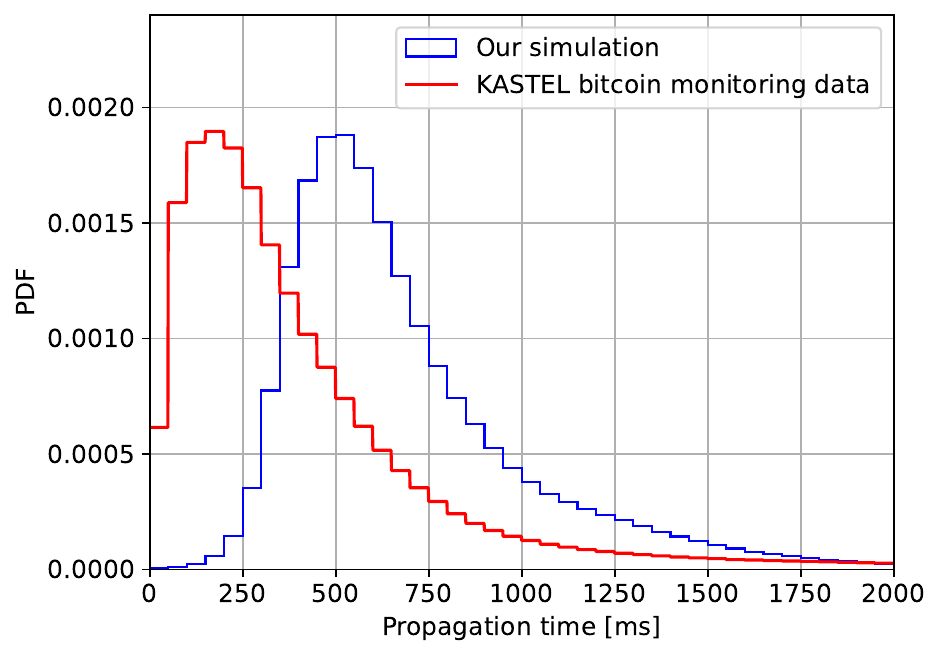}
	\caption{Comparison of the block arrival time in our environment and based on the measured data from KASTEL~\cite{kastel}.}
	\label{fig:calibration}
\end{figure}

\begin{figure*}[tb]
	\centering
	\begin{subfigure}[t]{0.24\textwidth}
		\captionsetup{justification=centering,margin=0.2cm}
		\includegraphics[width=\linewidth]{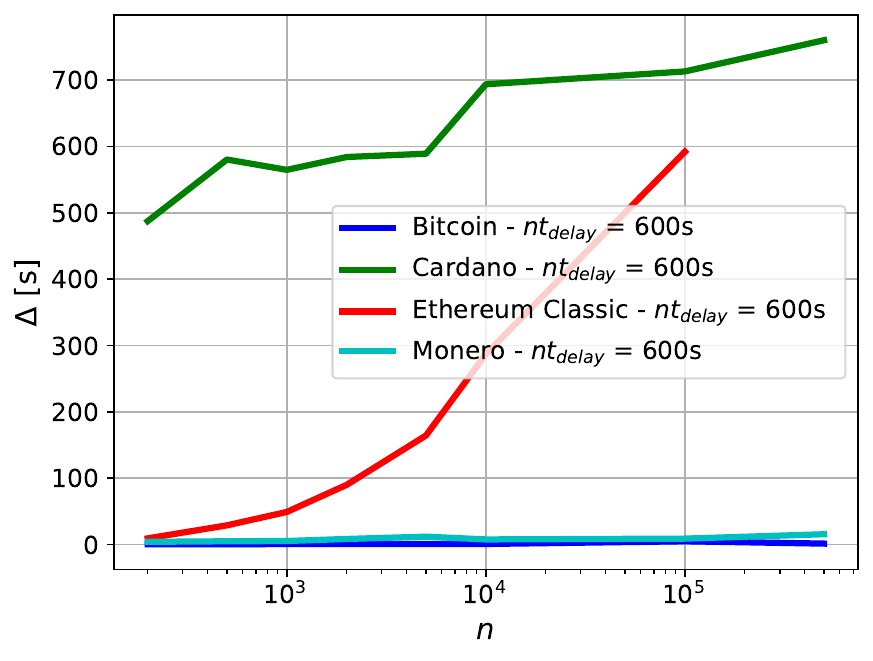}
		\caption{ $\maxdelay(\nodes)$ for $\EV=0.15$, $\ncs=0.1$}
		\label{fig:adv_0.15_0.1}
	\end{subfigure}
	\begin{subfigure}[t]{0.24\textwidth}
		\captionsetup{justification=centering,margin=0.2cm}
		\includegraphics[width=\linewidth]{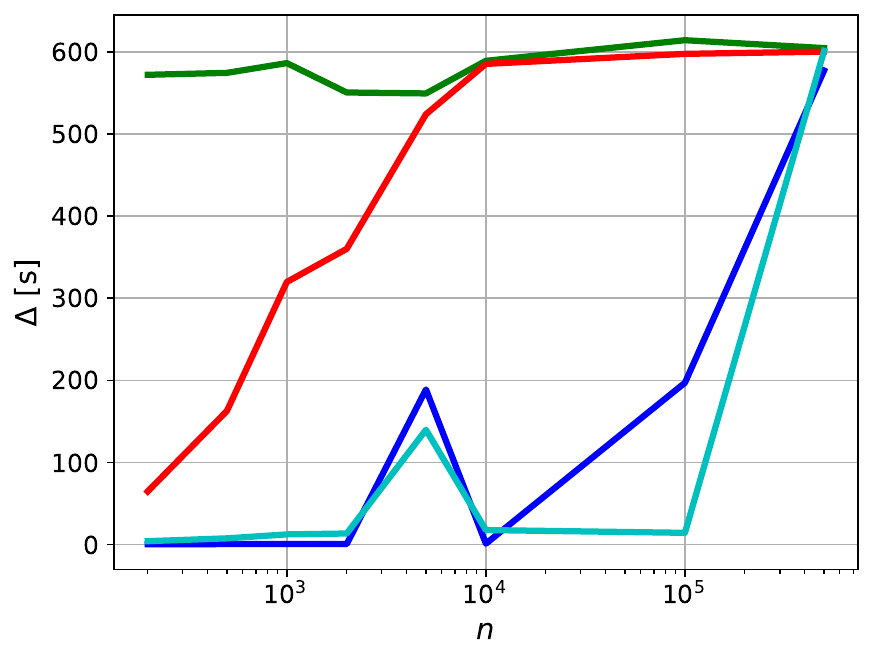}
		\caption{ $\maxdelay(\nodes)$ for $\EV=0.15$, $\ncs=0.5$}
		\label{fig:adv_0.15_0.5}
	\end{subfigure}
	\begin{subfigure}[t]{0.24\textwidth}
		\captionsetup{justification=centering,margin=0.2cm}
		\includegraphics[width=\linewidth]{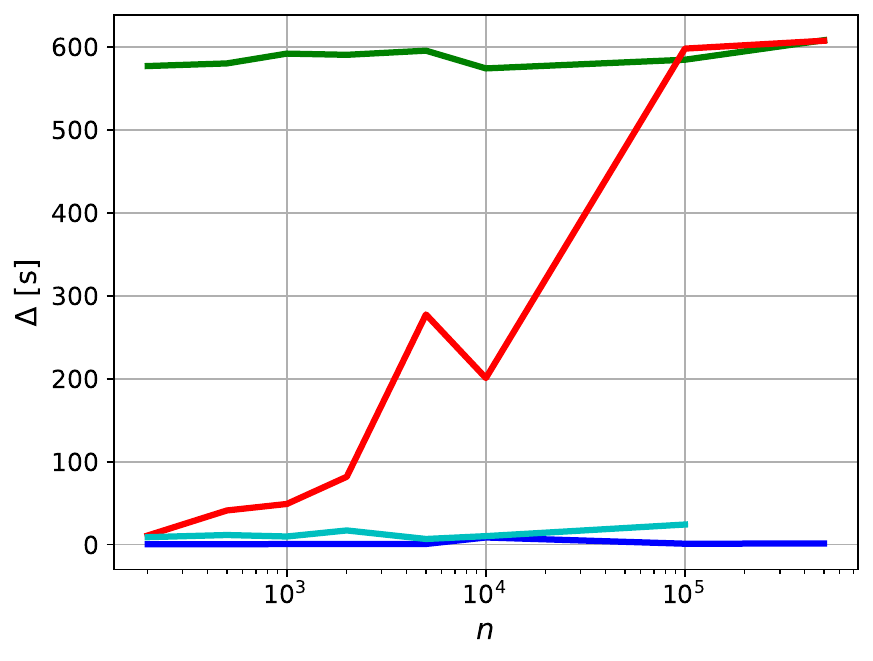}
		\caption{ $\maxdelay(\nodes)$ for $\EV=0.3$, $\ncs=0.1$}
		\label{fig:adv_0.3_0.1}
	\end{subfigure}
	\begin{subfigure}[t]{0.24\textwidth}
		\captionsetup{justification=centering,margin=0.2cm}
		\includegraphics[width=\linewidth]{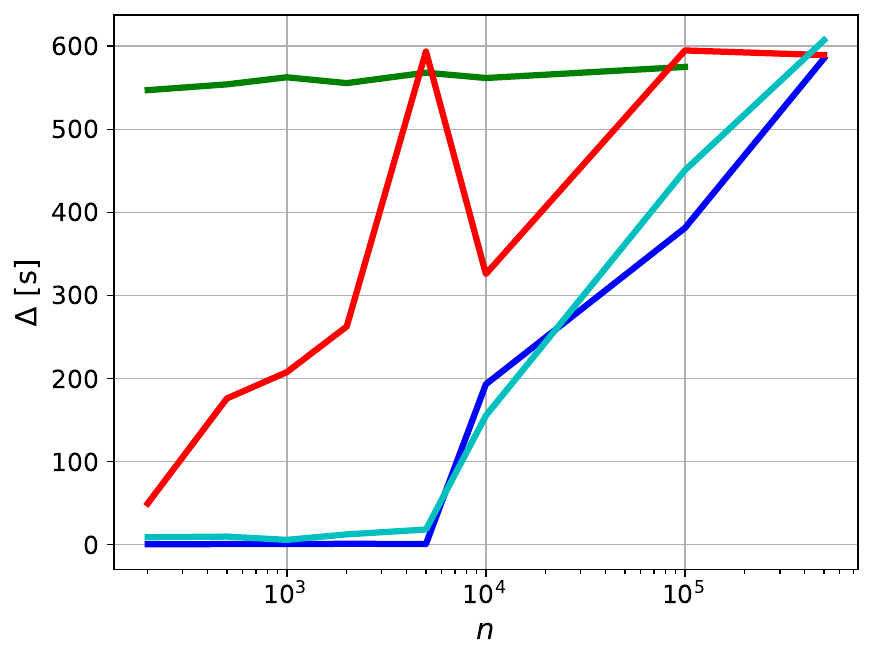}
		\caption{ $\maxdelay(\nodes)$ for $\EV=0.3$, $\ncs=0.5$}
		\label{fig:adv_0.3_0.5}
	\end{subfigure}
	\caption{Impact of network-layer attacks on $\maxdelay$ assuming a uniform distribution of block generation resources.}
	\label{fig:exp1_adversary}
\end{figure*}

\section{Parameter Calibration}\label{ap:calibration}
To ensure a realistic setup of the simulator, we calibrated the network's bandwidth and latency in such a way that it yields comparable block propagation times and stale block rate when compared to the KASTEL dataset~\cite{DBLP:conf/p2p/DeckerW13, kastel}.

The KASTEL dataset contains real-world measurements derived from the Bitcoin network, 
The measurements include the distribution of \textit{inv} messages' arrival times received by a monitoring node for each block sent in the network. 
This latter node was connected to 7000-9000 nodes in the Bitcoin network. 
Delays between a block's first and last received \textit{inv} messages were measured and counted in intervals of one millisecond each. In our work, we relied on the KASTEL datasets measured between 
January and October 2022 when calibrating our simulator.

To calibrate the parameters used in our simulator, we run a grid search to evaluate various parameter configurations and compare our simulator's output against the measured dataset.  
The overall aim of our grid search is to evaluate our simulator using various parameter combinations to yield a configuration that mimics the overall performance measured by the KASTEL dataset. 
To measure the conformity between our simulator's output and the real-world data, we compute the Mean Squared Error (MSE) over the distribution of propagation times.
By minimizing the MSE score, we calibrate our simulator to match the behavior of the real-world Bitcoin blockchain.     
The parameters that we configured in our grid search are as follows: 
\begin{description}
	\item[Upload bandwith] The upload bandwidth is a matrix that determines the distribution of upload bandwidth from nodes in one region to nodes in another region. 
	The unit of this parameter is \textit{bits per seconds}. 
	\item[Download bandwith] Similar to the upload bandwidth, the download bandwidth denotes the region-wise distribution of nodes' download speed.  
	\item[Latency] The latency is also defined as a matrix for latencies between pairs of regions and is defined in the unit of \textit{milliseconds}. 
\end{description}

\cref{fig:calibration} depicts the probability density function (PDF) derived from the KASTEL dataset~\cite{kastel} (in red) and the PDF derived from our simulator (in blue). 
As can be seen, our results follow a similar distribution, shifted to the right.
When comparing the mean propagation delay, our simulator yields 0.704 seconds, while the collected dataset compound to 0.783 seconds.
The main rational behind the shifted PDF is due to the starting point of the measurements: whereas we are able to know the exact time when a block is mined and when it is received by the different nodes, the dataset collected by \cite{kastel} can only setup the initial time to be the time when they receive the first inv for a given block.
However, we stress that both PDFs show a similar shape for the block propagation.

The results of our calibration on latency and throughput are depicted in Table~\ref{tab:latency} and Table~\ref{tab:throughput} respectively. The regions considered are \textbf{N}orth \textbf{A}merica, \textbf{Eu}rope, \textbf{S}outh \textbf{A}merica, \textbf{A}sia \textbf{S}outh, \textbf{A}sia \textbf{P}acific, \textbf{Jap}an, and \textbf{Aus}tria.
\begin{table}
	\centering
\begin{tabular}{c|ccccccc}
	\textbf{Latency [ms]} & NA & EU & SA & AS& AP & JAP & AUS\\
	\hline
NA & 24 & 70 & 79 & 115 & 140 & 79 & 122 \\
EU & 70 & 12 & 161 & 129 & 97 & 152 & 140 \\
SA & 79 & 128 & 30 & 195 & 207 & 176 & 201 \\
AP & 122 & 129 & 195 & 24 & 54 & 24 & 129 \\
AS& 122 & 97 & 207 & 54 & 36 & 42 & 85 \\
JAP & 103 & 152 & 176 & 48 & 42 & 18 & 183 \\
AUS & 134 & 158 & 201 & 128 & 85 & 192 & 30 \\
\end{tabular}

\caption{Latency between the origin region (rows) and the target region (columns) of 2 nodes, in milliseconds}
\label{tab:latency}
\end{table}

\begin{table}
	\centering
	\begin{tabular}{|l|lllllll|}
		\hline
		 Region: & NA & EU & SA & AS& AP & JAP & AUS\\
		\hline
		Upload [MB/s]: &6.8& 4.8 & 2.7& 4.8&	13.7& 6.8& 2.4\\
		\hline
		
	\end{tabular}
	\caption{Maximum upload throughput of each region in MB/s}
	\label{tab:throughput}
\end{table}

\section{Impact of Gossip Protocols}\label{ap:gossip_protcols}

We now evaluate in details the impact of the gossip layer on the measured delays. To do so, we consider a Bitcoin deployment and vary the gossip protocol to one of the following: \adv, \push, and \hybrid block propagation as well as \cbr in low-bandwidth mode.  In our implementation, we consider an adversary that can compromise nodes with probability $\EV$, such that on average $\EV \cdot \nodes$ are adversarially corrupted to execute network attacks. 
Each such malicious node delays messages to a fraction \ncs of its neighbors by artificially replying after \nct seconds. We vary each parameter as follow: $\EV \in \{15\%, 30\%\}$, $\ncs\in \{10\%,50\%\}$, and $\nct\in \{10m\}$, forming 4 distinct parameter combinations to test the impact of an attack on the network's delay and the resilience of the different gossip protocols to attacks.

\cref{fig:adv_0.15_0.1} shows that the \adv gossip protocol of Cardano is sensitive even to the weakest form (small values of $\EV$ and \ncs) of network-layer attacks, increasing the \maxdelay to 490s with $\nct=10$m for small values of \nodes. When the network size increases, the \maxdelay increases to 580s, which correspond to an increase of 20\% in latency.
The large \maxdelay measured even for small values of \nodes is mostly due to the fact that nodes that receive \inv messages from adversarial nodes will need to wait at least $min\{\nct,1200\}$ to either receive the block, or trigger a timeout and request the block from another node~\cite{DBLP:conf/ccs/GervaisRKC15}. 

Contrary to the \adv protocol that showcase high sensitivity to network-layer attacks, the \push and low-bandwidth \cbr mechanisms of Monero and Bitcoin seem to offer better robustness. This is due to the fact that nodes that are connected to at least one honest node will receive the block through a direct push path. In such a case, the network delay's impact is limited, hence the relatively small difference in performance, as depicted in \cref{fig:adv_0.15_0.1,fig:adv_0.3_0.1}, with the \maxdelay of 53s, much smaller than the delays incurred by the other two gossip protocols. 
Nonetheless, if some honest nodes are connected only to malicious nodes (cf. \cref{fig:adv_0.15_0.5,fig:adv_0.3_0.5}), the delays witnessed by both the \push and \cbr mechanism considerably increase up to 600s.

Finally, the \textit{hybrid push} protocol seems to establish a tradeoff between \adv and the direct push mode. 
While this gossip protocol achieves decent performance for most parameters when $n\leq10^3$, the performance degrades quickly when the network grows. In the case of $\EV=0.15$ and $\ncs=0.5$ shown in \cref{fig:adv_0.15_0.5} when $n=2000$, \maxdelay increases to 371s for $\nct=600$s, and when increasing \nodes to $10^5$, \maxdelay reaches 588s.

Overall, our findings in the adversarial setting confirm our intuition that \adv generally exhibits the highest \maxdelay, whereas \push and the low bandwidth mode of \cbr demonstrates a high level of consistency, with only a decline in performance in the worst cases ($\EV=0.3,~\ncs\geq0.5$). These findings suggest that \push and the low bandwidth mode of \cbr may be more reliable when taking network adversaries into account.

\newpage 
\section{Meta-Review}

The following meta-review was prepared by the program committee for the 2024
IEEE Symposium on Security and Privacy (S\&P) as part of the review process as
detailed in the call for papers.

\subsection{Summary}
This paper extends models of Nakamoto-style consensus to explicitly consider the likelihood of compromise of miners. The paper uses this more in-depth model to show that, contrary to prior results, smaller networks are not always more secure. Specifically, networks become more secure as they grow, until they become too large for their gossip protocol to transmit blocks fast enough, beyond which point their security begins to degrade. The paper investigates the relative impacts of gossip protocols of various practical deployments, and shows that the choice makes a big difference.

\subsection{Scientific Contributions}
\begin{itemize}
\item Independent Confirmation of Important Results with Limited Prior Research
\item Provides a Valuable Step Forward in an Established Field
\end{itemize}

\subsection{Reasons for Acceptance}
\begin{enumerate}
\item The paper combines analytical and empirical methods to better understand the impacts of sub-protocols (e.g., the gossip network) and network size on network security.
\item This paper provides a valuable step toward building more accurate models of the security of Nakamoto consensus.
\item A key takeaway from the paper is understanding the differences in behavior wrt latency and number of nodes for different protocols that all rely on Nakamoto consensus (or variants) but with potentially different P2P networks.
\item Much of the paper is well-written and easy to follow.
\end{enumerate}

\end{document}